%% file: main.tex
\documentclass[11pt]{article}

\usepackage[margin=1in]{geometry}
\usepackage{rpmacros}
\usepackage[T1]{fontenc}
\usepackage[usenames,dvipsnames]{color}
\usepackage{stmaryrd}
\usepackage{lineno}
\usepackage{xfrac}
\usepackage{mathpazo}
\usepackage{todonotes}
\usepackage{setspace}
\usepackage{nicefrac}
\usepackage{float}
\usepackage{scrextend}
\usepackage[linesnumbered,ruled]{algorithm2e}
\allowdisplaybreaks

\usepackage{fullpage}
\usepackage[utf8]{inputenc}

\usepackage{color}
\usepackage[
pagebackref, 
pdfstartview=FitH,pdfpagemode=UseNone,colorlinks=true,citecolor=blue,linkcolor=blue]{hyperref}
\hypersetup{
  linkcolor=[rgb]{0.3,0.3,0.6},
  citecolor=[rgb]{0.2, 0.6, 0.2},
  urlcolor=[rgb]{0.6, 0.2, 0.2}
}
\usepackage[nameinlink,capitalise]{cleveref}

\usepackage[title]{appendix}
\crefname{appendix}{Appendix}{Appendices}
\Crefname{appendix}{Appendix}{Appendices}

\input{thmmacros}

\numberwithin{algocf}{section}        

\makeatletter
\let\c@algocf\c@theorem

\makeatother

\usetikzlibrary{decorations.pathreplacing,arrows}

\newif\ifnote
\notefalse
\ifnote
\newcommand{\LPnote}[1]{\textcolor{BrickRed}{\guillemotleft PH: #1\guillemotright}}
\newcommand{\MKnote}[1]{\textcolor{Purple}{\guillemotleft MK: #1\guillemotright}}
\newcommand{\MSnote}[1]{\textcolor{OliveGreen}{\guillemotleft MS: #1\guillemotright}}
\newcommand{\LPnote}[1]{\textcolor{NavyBlue}{\guillemotleft SB: #1\guillemotright}}
\else
\newcommand{\LPnote}[1]{}
\newcommand{\MKnote}[1]{}
\newcommand{\MSnote}[1]{}
\fi

\newcommand{\iext}{\Vec{i}_\text{ext}}

\newcommand{\Lap}{\Vec{L}}
\newcommand{\D}{\Vec{D}}
\newcommand{\A}{\Vec{A}}
\renewcommand{\vec}[1]{{\mathbf #1}}

\let\epsilon\varepsilon

\newcommand{\vp}{\vecp}

\newcommand{\vf}{\vecf}
\newcommand{\vx}{\vecx}

\title{Bounded-Independence Sampling of Edges\\ for Combinatorial Graph Properties}
\author{Aaron Putterman\thanks{Harvard University, Cambridge, Massachusetts, USA. Supported in part by the Simons Investigator awards of Madhu Sudan and Salil Vadhan, and AFOSR award FA9550-25-1-0112. Email: \texttt{aputterman@g.harvard.edu}.}
\and
{ 
{Salil Vadhan\thanks{Harvard University, Cambridge, Massachusetts, USA. Supported by a Simons Investigators Award. Email: \texttt{salil\_vadhan@harvard.edu}.}}
}
\and
{ 
{Vadim Zaripov\thanks{Harvard University, Cambridge, Massachusetts, USA. Email: \texttt{vadimzaripov@college.harvard.edu}. Part of this work was done in partial fulfillment of the author’s undergraduate thesis at Harvard University.}}
}
}

\date{\today}

\begin{document}

\maketitle

\begin{abstract}

Random subsampling of edges is a commonly employed technique in graph algorithms, underlying a vast array of modern algorithmic breakthroughs. Unfortunately, using this technique often leads to randomized algorithms with no clear path to derandomization because the analyses rely on a union bound on exponentially many events. In this work, we revisit this goal of \emph{derandomizing} randomized sampling in graphs.

We give several results related to bounded-independence edge subsampling, and in the process of doing so, generalize several of the results of Alon and Nussboim (FOCS 2008), who studied bounded-independence analogues of random graphs (which can be viewed as edge subsamples of the complete graph). Most notably, we show that in graphs with $m$ edges:
\begin{enumerate}
    \item $O(\log(m))$-wise independence suffices for preserving connectivity when sampling at rate $1/2$ in a graph with minimum cut $\geq \kappa \log(m)$ with probability $1 - 1/\mathrm{poly}(m)$ (for a sufficiently large constant $\kappa$).
    \item $O(\log(m))$-wise $(1/\mathrm{poly}(m))$-almost independence suffices for ensuring cycle-freeness when sampling at rate $1/2$ in a graph with minimum cycle length $\geq \kappa \log(m)$ with probability $1 - 1/\mathrm{poly}(m)$ (for a sufficiently large constant $\kappa$).
    \item If we relax to \emph{arbitrary} distributions, we show there is an explicit distribution on $\{0, 1\}^m$ with marginals $\leq 1/2$ generated using $O(\log(m)\log\log(m))$ random bits such that in a graph with minimum cut $\geq \kappa \log(m)$ (for a sufficiently large constant $\kappa$), a sample from the distribution is still connected with probability $1- 1/\poly(m)$.
\end{enumerate}

To demonstrate the utility of our results, we revisit the classic problem of using parallel algorithms to find graphic matroid bases, first studied in the work of Karp, Upfal, and Wigderson (FOCS 1985). In this regime, we show that the optimal algorithms of Khanna, Putterman, and Song (arxiv 2025) can be \emph{explicitly} derandomized while maintaining near-optimality. 
\end{abstract}

\vfill

\pagenumbering{gobble}

\pagebreak

{\small\tableofcontents}

\pagebreak
\pagenumbering{arabic}
	
\section{Introduction}

\subsection{Background}

Suppose that one wishes to run an algorithm on an Erdős–Rényi random graph $\mathcal{G}(n, p)$. There are exponentially many graphs in the support of $\mathcal{G}(n, p)$, so even storing a graph requires resources polynomial in $n$. A common approach for derandomizing such an algorithm is to replace $\mathcal{G}(n, p)$ with some random looking distribution $\mathcal{G}_n$ whose support is much smaller (for example, $\poly(n)$), so that all the properties of random graphs on which the algorithm relies are still preserved. This line of work was initiated by Goldreich, Goldwasser, and Nussboim \cite{GGN03}, who studied pseudorandom graphs that are indistinguishable from $\mathcal{G}(n, p)$ by any oracle machine that makes $\poly(\log n)$ adjacency queries to a graph. Goldreich et al. considered several properties of random graphs, such as connectedness and Hamiltonicity, and presented pseudorandom graph constructions that preserve these properties while having support polynomial in $n$.

Alon and Nussboim \cite{alon2008k} studied the specific strategy of using $k$-wise independence for sampling a pseudorandom graph. 
These can be sampled with support size $n^{\Theta(k)}$ and stored and accessed with $\widetilde{O}(k\log n)$ space and time. Similarly to \cite{GGN03}, \cite{alon2008k} show that many typical properties and thresholds that are enjoyed by truly random graphs \emph{still} hold true with $k$-wise independent sampling for $k = O(\log n)$: for instance, edge connectivity, vertex connectivity, jumbledness, and the existence of matchings all mimic the behavior of truly independent sampling. 

We continue this line of study by focusing on random subsampling of edges in graphs: instead of looking at $\mathcal{G}(n, p)$, we start with a given graph and randomly take each of its edges in our sample with some fixed probability, getting a sparser subgraph. This is a fundamental technique which underlies a huge portion of modern graph algorithms, including sparsification and its many applications \cite{BK96, ST11}, coloring \cite{ACK19}, vast arrays of sublinear algorithms \cite{mcgregor2014graph}, and many more. However, for these graph algorithms, the gold standard is often the design of \emph{deterministic} algorithms, thereby bypassing any uncertainty in the algorithm's output. Motivated by this, the focus of one line of study, including this work, is on the ability to \emph{derandomize} graph sampling.
 
Subsequent works have studied the feasibility of $k$-wise independent subsampling in arbitrary graphs: as an example, the work of Doron, Murtagh, Vadhan, and Zuckerman \cite{Doron2020SpectralSV} studied the ability of $k$-wise independent sampling to produce \emph{spectral sparsifiers} of graphs, again showing that mild independence (in fact, $k = O(\log(n))$) suffices for preserving the graph's spectrum. 

However, these aforementioned results rely on a shared insight: namely that bounded independence sampling suffices in good spectral expanders for concentration of the \emph{eigenvalues} of a graph's Laplacian. Alon and Nussboim \cite{alon2008k} leverage the fact that a complete graph is a good expander, and therefore has eigenvalues well-separated from $0$, and Doron et al. \cite{Doron2020SpectralSV} assign different sampling probabilities to each edge proportional to their effective resistances, which was a procedure known to preserve graph's eigenvalues with fully independent sampling \cite{SS11}.

Thus, in our work we seek to understand:
\begin{quote}
    \emph{Does $k$-wise independent subsampling for small $k$ also preserve \emph{combinatorial} properties of graphs that are not good expanders?}
\end{quote}
By "combinatorial" properties, we mean those where we do not know how to use the spectrum directly to analyze fully independent subsampling.

As we shall see, we show that for certain properties the answer is strongly affirmative. In the following section, we discuss these results in more detail and subsequently explain how they lend themselves to new explicit derandomized graph algorithms.

\subsection{Our Graph Sampling Results}\label{sec:graphSamplingDiscussion}

As an illustration of the challenges we face, consider a graph $G = (V, E)$ on $n$ vertices with minimum cut $\lambda \approx 100 \log(n)$. Karger's \cite{Kar93} cut-counting bound establishes that in such a graph, the number of cuts of size $\leq \alpha \lambda$ is at most $n^{2 \alpha}$. An immediate corollary of this is that in such a graph $G$, if one samples each edge independently with probability $1/2$, then the resulting graph is connected with probability at least $1 - 1 / \mathrm{poly}(n)$. To see why, we bin our cuts to those of size in the interval $[\lambda, 2 \lambda)$, $[2\lambda, 4\lambda)$, etc. Each cut in the bin $[\alpha\lambda, 2\alpha\lambda)$ has at least one edge in the surviving sample of edges with probability $1 - 1 / 2^{\alpha \lambda} = 1 - 1 / n^{100 \alpha}$. We can then take a union bound over all $\leq n^{4\alpha}$ cuts in this bin and then a further union bound over all $\leq n$ bins to conclude that \emph{every} cut has at least one surviving edge in the sample with high probability, and thus the graph is connected.

As mentioned above, a natural approach to derandomizing the above property is to use $k$-wise independent sampling. Unfortunately, if we revisit the argument above, we can see that an essential step in arguing that the graph is connected under random subsampling is to take a union bound over an \emph{exponential} number of cuts in the graph. This relies on having ``success'' probabilities for individual cuts that are as large as $1 - 2^{- (n-1)}$, since there are $2^{n-1}$ cuts in the graph. In general, no pseudorandom algorithm with a sample space smaller than $2^{n-1}$ can guarantee such high success probabilities, since for such an algorithm, the individual probabilities are multiples of $2^{-(n-1)}$.

Furthermore, a graph having a minimum cut of size $100 \log(n)$ provides \emph{no meaningful bound} on its expansion or the effective resistances of its underlying edges, meaning that the techniques of \cite{alon2008k} and \cite{Doron2020SpectralSV} do not directly lend themselves to this setting. Thus, it may seem that we are unable to conclude this relatively benign connectivity property under bounded-independence sampling.

Despite the appearance of the above barriers, we show that for $k = O(\log(n))$, $k$-wise independent sampling still yields connected graphs with high probability.

\begin{theorem}\label{thm:introConnected}
    Let $G = (V, E)$ be a graph with $m$ edges and minimum cut $\lambda \geq \kappa \log(m)$ for some absolute constant $\kappa$. Let $\widetilde{G}$ denote the result of subsampling edges of $G$ in a $2\kappa \log(m)$-wise independent manner with marginals $1/2$. Then, $\widetilde{G}$ is connected with probability $\geq 1 - 1 / \mathrm{poly}(m)$.
\end{theorem}

This result serves as a strong generalization of some of the results of \cite{alon2008k}. Indeed, \cite{alon2008k} studies which properties of \emph{the complete graph} are preserved under $k$-wise independent subsampling. Naturally, working with such a structured graph enables simpler analysis. Nevertheless, in our much less structured regime where graphs only have a lower bound on their minimum cut, we are still able to show that $O(\log(m))$-wise independent sampling suffices for preserving connectivity.

Sampling from a $O(\log m)$-wise independent distribution requires $\Theta(\log^2(m))$ random bits; indeed the support must be of size $2^{\Omega(\log^2(m))}$ \cite{AGM03}. Thus, as an alternative, we construct another sampling procedure that maintains connectivity while having almost-polynomial  support size:

\begin{theorem}\label{thm:introConnected2}
There exists an explicit distribution $Y$ on $\{0, 1\}^m$ that can be sampled using $O(\log m\cdot \log\log m)$ random bits and has marginals $\leq 1/2$ such that the following holds. For every graph $G = (V, E)$ with $m$ edges and minimum cut $\geq \kappa\log(m)$ for some absolute constant $\kappa$, subsampling the edges of $G$ using $Y$ yields a connected subgraph with probability $1 - 1/\poly(m)$.
\end{theorem}

In a similar manner, we also study \emph{cycle-free-ness} when randomly sampling edges with large girth. In more detail, consider a graph $G = (V, E)$ whose shortest cycle is of length $\geq 100 \log(n) = \lambda$. A bound due to Subramaniam \cite{Sub95} (and rediscovered in several other works, including \cite{FGT16}) states that in such a graph, the number of cycles of length $\leq \alpha \cdot \lambda$, is bounded by $n^{2 \alpha}$. Thus, if one samples the edges in this graph uniformly at rate $1/2$, the exact same argument as in the cut setting above will show that the resulting sample of edges is \emph{cycle free} with high probability. However, this argument again relies on exponential concentration in the tail, and thus does not immediately translate into a proof of the same property holding under bounded independence. Despite this, we are able to show that bounded independence suffices:

\begin{theorem}\label{thm:introCycleFree}
    Let $G = (V, E)$ be a graph with $m$ edges and minimum cycle length $\lambda \geq \kappa \log(m)$ for some absolute constant $\kappa$. Let $\widetilde{G}$ denote the result of subsampling edges of $G$ in a $(1/m^{200})$-almost $2\kappa \log(m)$-wise independent manner with marginals $1/2$. Then, $\widetilde{G}$ is cycle free with probability $\geq 1 - 1 / \mathrm{poly}(m)$.
\end{theorem}

Unlike \cref{thm:introConnected}, this theorem allows for the use of an \emph{almost} $O(\log m)$-wise independendent distribution, which can be sampled using $O(\log m)$ random bits.

Together, these results show that bounded-independence sampling \emph{does suffice} for preserving basic combinatorial properties in weakly-structured graphs (i.e., non-expanders). As we elaborate upon more in the Technical Overview (\cref{sec:techOverview}), our results rely on very careful analyses of graph structure.  Furthermore, this improved understanding of graph sampling under bounded-independence immediately lends itself to better deterministic parallel algorithms for finding matroid bases.

\subsection{Applications to Parallel Matroid Algorithms}

\subsubsection{Matroid Background}

Matroids are fundamental and ubiquitous objects in combinatorial optimization, generalizing a number of basic structures like spanning forests in graphs, and bases of vector spaces. More formally, a matroid is modelled as a set system $\mathcal{M} = (E, \mathcal{I})$, where $E$ is the \emph{ground set}, consisting of $m$ elements, and $\mathcal{I} \subseteq 2^E$ is the \emph{set of independent sets}, satisfying:
\begin{enumerate}
    \item $\emptyset \in \mathcal{I}$ (non-empty property).
    \item If $S \in \mathcal{I}$ and $S' \subseteq S$, then $S \in \mathcal{I}$ (downward closure property).
    \item If $S, T\in \mathcal{I}$ and $|S| > |T|$, then $\exists e \in S\setminus T: T \cup \{e\} \in \mathcal{I}$ (exchange property).
\end{enumerate}

Within the study of matroids, a critical notion is that of a \emph{basis}, or namely a set $S \in \mathcal{I}$ that is \emph{maximal} under inclusion. Note that because of the exchange property above, we know that for a given matroid $\mathcal{M} = (E, \mathcal{I})$, all bases of $\mathcal{M}$ are of the same size.

For instance, one can take the ground set $E$ to be the set of edges in a graph $G = (V, E)$, and let a set of edges $S \subseteq E$ be independent if and only if $S$ is \emph{cycle-free} in the graph $G$. This yields the family of so-called \emph{graphic matroids}, with a maximal independent set (i.e., basis) being a spanning forest of $G$.

\subsubsection{Parallel Basis Finding}

Despite their prevalence in combinatorial optimization (see e.g., \cite{edmonds1979matroid, jensen1982complexity, gabow1984efficient, kleinberg2012matroid, balkanski2019optimal, blikstad2023fast}), our understanding of matroids is incomplete, even in basic directions. One such direction is in the study of \emph{parallel algorithms} for finding matroid bases. In this model, first introduced by Karp, Upfal, and Wigderson \cite{KUW85, KUW88}, one wishes to find a basis of some matroid $\mathcal{M} = (E, \mathcal{I})$. However, because the number of matroids on $n$ elements grows as $2^{2^n}$ \cite{BPV15}, it is too expensive to store the entire description of the matroid. Instead, the algorithm is given access to the matroid only in the form of an \emph{independence oracle}; i.e., a function $\mathrm{Ind}: 2^E \rightarrow \{0,1\}$, such that $\mathrm{Ind}(S) = \mathbf{1}[S \in \mathcal{I}]$. Under this form of access, the works of Karp, Upfal, and Wigderson \cite{KUW85, KUW88} asked:

\begin{quote}
    \emph{How many adaptive rounds of independence queries are required to find a basis of an arbitrary matroid, when each round is allowed only $\mathrm{poly}(n)$ many queries to the independence oracle?}
\end{quote}

In the seminal work \cite{KUW85}, it was proven that in the setting of arbitrary matroids, there is a \emph{deterministic} algorithm which finds a basis of any matroid in just $O(\sqrt{n})$ rounds, and that there is a lower bound of $\Omega(n^{1/3})$ many rounds for this task. Despite providing a narrow gap for the parallel complexity of this basic problem, the ensuing four decades saw no improvements to these bounds. Only recently, in the work of Khanna, Putterman, and Song \cite{khanna2025parallel}, were these bounds improved, with the establishment of an $\widetilde{O}(n^{7/15})$ round \emph{randomized} algorithm for finding bases in arbitrary matroids. Nevertheless, the true complexity of this problem is yet to be established.

Due to the complex structure of arbitrary matroids, the work of \cite{KUW85} proposed studying the parallel complexity of \emph{concrete classes} of matroids. Among these, the most notable class studied by \cite{KUW85} is the class of \emph{graphic matroids}:

\begin{definition}
    Given a graph $G = (V, E)$, the \emph{graphic matroid} induced by $G$ is the matroid $\mathcal{M}= (E, \mathcal{I})$ such that for a set $S \subseteq E$, $S \in \mathcal{I}$ if and only if $S$ is cycle-free.
\end{definition}

For this class of matroids, \cite{KUW85} provided deterministic algorithms that, on graphs with $m$ edges, achieve either (a) for any constant $d \in \mathbb{N}$, $m^{1 / d}$ rounds and $m^{O(d)}$ queries per round, or (b) $O(\log(m))$ rounds with $m^{O(\log(m))}$ queries per round. This was later improved in the work of Khanna, Putterman, and Song \cite{khanna2025optimal}, which provided an algorithm that finds bases of graphic matroids with only $\mathrm{poly}(m)$ queries per round and $O(\log(m))$ many rounds. 

Note that, as mentioned in their work, the algorithms of \cite{khanna2025optimal} can be derandomized \emph{non-explicitly}, in an argument akin to Adleman's proof of $\mathbf{BPP} \subseteq \mathbf{P} / \mathbf{poly}$ \cite{adleman1978two}. This of course answers the natural question of the deterministic ``query complexity'' of finding a matroid basis, but does not yield an efficiently implementable uniform algorithm. 

\subsubsection{Our Results}

Our first result uses \cref{thm:introCycleFree} to deterministically find bases in \emph{graphic matroids}. In this regime, we show the following:

\begin{theorem}\label{thm:graphicMain}
    There is a uniform, deterministic algorithm that, for any graphic matroid on $m$ elements, finds a basis in $O(\log m\cdot \log\log m)$ many rounds of $\mathrm{poly}(m)$ queries.
\end{theorem}

This theorem shows that the result of \cite{khanna2025optimal} \emph{can} be explicitly derandomized with only a minor increase in the number of rounds (from $\log(m)$ to $\log(m) \log\log(m)$), while still retaining a polynomial number of queries per round. The proof of this result relies intimately on \cref{thm:introCycleFree} (in fact, even on a strengthening of \cref{thm:introCycleFree}).

Next, as an application of \cref{thm:introConnected2}, we show that basis finding in \emph{cographic} matroids can also be derandomized. A set of edges $S \subseteq E$ in a cographic matroid for a graph $G$ is defined to be independent if and only if $S$ does not completely contain any non-empty cut $C$ of $G$, and so a basis of a cographic matroid is a set of edges that is the \emph{complement} of a spanning forest.

\begin{theorem}\label{thm:cographicMain}
    There is a uniform, deterministic algorithm that, for any cographic matroid on $m$ elements, finds a basis in only $O(\log(m) \log\log(m))$ many rounds of $m^{O(\log\log(m))}$ queries.
\end{theorem}

Note that, prior to this work, there was no non-trivial uniform, deterministic algorithm for basis finding in cographic matroids. The best algorithm was only known to use $O(\sqrt{m})$ rounds and $\mathrm{poly}(m)$ queries per round - due to \cite{KUW85} for \emph{arbitrary matroids} (thus including cographic ones as a special case). We view \cref{thm:cographicMain} as further evidence that one can likely obtain $\mathrm{polylog}(m)$ rounds and $\mathrm{poly}(m)$ queries, though we leave this as an open question.

As we elaborate upon in the next section, both \cref{thm:graphicMain} and \cref{thm:cographicMain} rely on a new algorithmic framework for finding bases of matroids which, unlike the algorithms of \cite{khanna2025optimal}, uses both deletion and contraction of elements. However, to remove the randomness in this algorithmic framework, our results rely on the concrete derandomization results for basic graph sampling tasks that we discussed above.

\subsection{Technical Overview}\label{sec:techOverview}

In this overview, we discuss some of the intuition that underlies both our derandomized graph sampling results as well as our improved explicit deterministic matroid algorithms. 

\subsubsection{Derandomizing Graph Sampling}

We start by providing intuition for \cref{thm:introConnected}. Before doing so, we briefly recap the discussion in \cref{sec:graphSamplingDiscussion}. In this setting, we are given a graph $G = (V, E)$ with minimum cut $\lambda$ of size $\geq \kappa \log(m)$ for some constant $\kappa$. Our goal is to subsample the edges of this graph at rate $1/2$ using a \emph{bounded-independence} distribution while still ensuring that the resulting sampled graph is connected with high probability. As mentioned in \cref{sec:graphSamplingDiscussion} however, a simple union bound over all of the cuts in the graph \emph{does not} suffice for arguing connectivity when using bounded-independence distributions. Indeed, there are $2^{n-1}$ cuts in a graph on $n$ vertices, and thus any union bound requires that most such cuts have at least one surviving edge with probability $1 - \Omega\left(2^{-n}\right)$; a probability which is too close to $1$ to be possible under $k$-wise independent sampling with $k = o(n)$. 

For comparison, \cite{alon2008k} and \cite{Doron2020SpectralSV} have studied bounded-independence sampling for small $k$ (generally, $k = O(\log(n))$) and showed that in these settings, many spectral properties of graphs which are good expanders can be preserved. Here, spectral properties of a graph $G = (V, E)$ refer to properties of the graph's Laplacian $L_G = \sum_{e= (u,v) \in E} w_e \cdot \chi_e \chi_e^T$, where $\chi_e \in \mathbb{R}^n$ is the vector such that $\chi_e = \mathbf{1}_u - \mathbf{1}_v$ when $e = (u,v)$. In this direction, \cite{Doron2020SpectralSV} provides a more general result than \cite{alon2008k}: indeed, it follows from \cite{Doron2020SpectralSV} that if the \emph{leverage scores} in a graph $G$ are all smaller than $\approx \frac{\epsilon^2}{2 \log(n)}$, then $O(\log(n))$-wise independent sampling with marginals $1/2$ yields a \emph{spectral sparsifier} of the graph $G$. In this context, the leverage score of an edge $e = (u,v)$ is defined as $\ell(u,v) = w_e \cdot \chi_e^T L_G^{\dagger} \chi_e$ and provides a measure of how important the edge $(u,v)$ is for preserving connectivity between $u$ and $v$. A spectral sparsifier of a graph $G$ is simply a subgraph $\widetilde{G}$ such that $(1 - \epsilon)L_G \preceq   L_{\widetilde{G}} \preceq (1 + \epsilon)L_G$, where $\preceq$ is the Loewner order on PSD matrices ($A \preceq B$ if $x^TAx \leq x^TBx$ for all $x$). If $\widetilde{G}$ is a spectral sparsifier of $G$, then this implies that for every $x \in \{0,1\}^n$, $x^T L_{\widetilde{G}}x \in (1 \pm \epsilon) x^T L_G x$, which means that $x^T L_{\widetilde{G}}x = 0$ if and only if $x^T L_G x = 0$. This means that $\widetilde{G}$ and $G$ will have the same number of connected components (since these equal the multiplicity of $0$ as an eigenvalue of the Laplacian), and in fact the connected components will be identical.

This provides a natural first attempt towards showing that $O(\log(m))$-wise independent sampling in our graph $G$ preserves connectivity. I.e., can we try to argue the \emph{stronger} property that the resulting sampled edges constitute a \emph{spectral sparsifier} of $G$?
Unfortunately, this turns out not to be the case. This is because a graph having a large minimum cut is very strongly \emph{not} sufficient for the graph to have small leverage scores (see \cite{FHHP11} for more discussion). Indeed, even if we tried sampling the edges of $G$ \emph{uniformly} at random, the resulting graph would not be guaranteed to be a spectral sparsifier with high probability. 

To overcome this, we rely on a key observation: the graph $G$ that we are working with so far is \emph{unweighted}. But, edge weights \emph{can} alter the behavior of the graph's Laplacian, and in turn, the leverage scores of the edges. Thus, one may wonder: is it possible to re-weight the edges of the graph such that all the leverage scores become smaller? Indeed, one of our key contributions is to show that this is true, and in doing so, we provide a new connection between leverage scores and a graph's minimum cut:

\begin{theorem}\label{thm:levScoreIntro}
Given an unweighted graph $G = (V, E)$ with minimum cut $c$, there exists a weighting $w: E\to \R_{\geq 0}$ such that for the weighted graph $G' = (V, E, w)$ and for each $e \in E$, the leverage score of $e$ in $G'$ is $\ell(e) \leq O(1/c)$.

Moreover, the converse is also true: if there exists a weighting such that all leverage scores are at most $1/c$, the minimal cut of the original graph is $\geq c$.
\end{theorem}

With this theorem in hand, we are then immediately able to show that connectivity is preserved whenever our graph $G$ has minimum cut $\geq \kappa \log(m)$: indeed, given the graph $G = (V, E)$, we implicitly analyze the \emph{re-weighted} graph $\hat{G}$. On $\hat{G}$, we sample each edge using our $O(\log(m))$-wise independent distribution, and are guaranteed via \cite{Doron2020SpectralSV} that the resulting sampled edges are a spectral sparsifier of $\hat{G}$, which in particular means that the sampled edges form a connected graph. Then, one must only observe that the sampling procedure on $G$ and $\hat{G}$ is identical, as both have the same set of edges.

Proving \cref{thm:levScoreIntro} is more subtle; our proof relies on a ``bounded effective resistance diameter'' decomposition theorem of \cite{DBLP:conf/innovations/AlevALG18}, along with a careful, recursive weight assignment scheme. We omit the details from the technical overview for brevity, and direct the reader to \cref{sec:leverageScoreSec} for more details.

In a separate direction, our proof of \cref{thm:introCycleFree} foregoes a spectral argument entirely. Indeed, in this setting where we must argue \emph{cycle-freeness}, we instead show that the existence of \emph{any cycles} in the graph can instead be modeled via the existence of \emph{short paths} in the graph. This allows us to transition from an event space of nearly-exponential size (i.e., the set of all cycles) to instead work with an event space of only polynomial size (pairs of vertices with short paths between them). We present this argument in \cref{sec:cycleFreeness}.

The proof \cref{thm:introConnected2} relies on different techniques. We construct our distribution $Y$ on $\{0, 1\}^m$ by combining $O(\log m)$ independent copies of any distribution $Y_0$ on $\{0, 1\}^{m\times \log \log m}$ that fools read-once DNF formulas on $\widetilde{O}(\log m)$ (out of the $m$ possible) variables. By the results of De, Etesami, Trevisan, and Tulsiani \cite{DETT10}, such a distribution $Y_0$ can be sampled using $O(\log \log m)$ random bits, meaning that we can generate $Y$ using $O(\log m\cdot \log\log m)$ random bits. Interestingly, the way we combine the copies of $Y_0$ to set $Y$ is also using a read-once DNF formula, namely the Tribes function.

Together, these results provide a robust toolkit for reasoning about graph properties under bounded-independence sampling when the underlying graph is not a good expander. As we shall see below, this has direct applications to derandomized matroid basis computation.

\subsubsection{Derandomized Matroid Basis Finding}

Once we have our improved derandomized graph sampling results, there are still barriers towards implementing better explicit parallel basis finding algorithms. To illustrate these barriers, we first revisit the algorithm of \cite{khanna2025optimal}.

Indeed, let us consider the setting of graphic matroids; recall that here there is an underlying graph $G = (V, E)$, and the goal of the algorithm is to find a \emph{spanning forest} of $G$. The algorithms in this setting only have access to an \emph{independence oracle}, meaning that the algorithm can query a set $S \subseteq E$ of edges, and the oracle reports whether the set $S$ has any cycles. The intuition for the algorithm of \cite{khanna2025optimal} is that each round of the algorithm \emph{deletes} more edges from the graph without altering the connected components of the graph. Then, after $O(\log(n))$ many rounds of computation, \cite{khanna2025optimal} shows that the graph has no cycles remaining, and thus the leftover edges must be a spanning forest of the original graph $G$. To implement this argument more carefully, \cite{khanna2025optimal} relies on two key invariants:
\begin{enumerate}
    \item After $i$ rounds of the algorithm, the remaining edges $E^{(i)} \subseteq E$ induce the same connected components as $E$.
    \item After $i$ rounds of the algorithm, there are no cycles in $E^{(i)}$ of length $\leq 1.01^{i}$.
\end{enumerate}
Clearly then, one can see that after $O(\log(n))$ rounds, the remaining edges constitute a spanning forest. 

To actually ensure these invariants hold, \cite{khanna2025optimal} relies on random sampling of edges. Indeed, in an iteration $i$ where the minimum cycle length in $E^{(i)}$ is $\lambda = 1.01^{i}$, \cite{khanna2025optimal} shows that \emph{uniformly random sampling} of edges can be used to isolate any near-minimum length cycle. Formally, \cite{khanna2025optimal} shows:
\begin{lemma}[\cite{khanna2025optimal}]
    Let $G = (V, E)$ be a graph with $m$ edges and minimum cycle length $\lambda$. Let $C \subseteq E$ be a cycle in $G$ of length $\leq 1.01 \lambda$. Now, let $\widetilde{G}$ be the result of uniformly randomly sampling the edges of $G$ at rate $p = \frac{1}{m^{100 / \lambda}}$. Then, $C \subseteq E$ is the \emph{unique} surviving cycle in $\widetilde{G}$ with probability at least $\frac{1}{\mathrm{poly}(m)}$.
\end{lemma}

With such a lemma in hand, \cite{khanna2025optimal} performs random sampling a large polynomial number times, ensuring that every single near-minimum length cycle is the unique surviving cycle in some sampled graph. This then enables \cite{khanna2025optimal} to recover the identities of \emph{all} near-minimum length cycles (and in so doing, also enables their removal without altering connectivity).

Our goal is to adopt this algorithmic framework, but to use bounded-independence sampling instead of truly random sampling. Ultimately, this then allows us to \emph{enumerate} all possible samples of edges, thereby obtaining explicit, deterministic results.

Unfortunately, directly using our derandomized graph sampling results \emph{does not} work inside this algorithmic framework. There are two primary reasons why:
\paragraph{1. Unique Cycle Survival} First, as currently stated, \cref{thm:introCycleFree} only governs the probability that a sampled graph is \emph{cycle-free}, \emph{not} the probability that there is a single, unique cycle which survives the bounded-independence sampling. Fortunately, \cref{thm:introCycleFree} can be modified to ensure unique cycle survival in the regime where the minimum cycle length is $\lambda = \Theta(\log(m))$. To see why, fix a single cycle $C$ of length $\leq 1.01 \lambda$. For this cycle, if one uses (almost) $k$-wise independent sampling with $k \geq 10 \lambda$, then $C$'s survival probability under bounded-independence sampling is in fact \emph{the same} as under a uniformly random distribution. Now, conditioned on $C$'s survival, we can re-invoke  \cref{thm:introCycleFree} on the remaining graph to ensure that \emph{no other cycles} survive sampling. This argument requires care, and appears in \cref{sec:uniqueSurvival}.

\paragraph{2. Large Cycle Lengths}  The second barrier is that, as stated, \cref{thm:introCycleFree} requires a minimum cycle length of at least $\kappa \log(m)$, along with sampling at rate $1/2$. Not only this, but in the previous barrier, the argument we just discussed relies on using $k$-wise (almost) independent sampling for $k \geq 10 \lambda$, where $\lambda$ is the current minimum cycle length. 

Unfortunately, the algorithmic framework of \cite{khanna2025optimal} requires a spectrum of different cycle lengths and sampling rates, starting with a minimum cycle length of $1$, and slowly increasing as edges are deleted between rounds. Thus, there are even rounds where the minimum cycle length $\lambda = \Omega(m)$, which in our above argument requires $\Omega(m)$-wise independent sampling, which is far too costly. 

To overcome this barrier, we employ a much more careful analysis than \cite{khanna2025optimal}: rather than continuing to remove short cycles until the minimum cycle length is $> m$, we instead invest $O(\log\log(m))$ rounds until the minimum cycle length becomes $\approx \kappa \log(m)$. Note that in these rounds where the minimum cycle length is $o(\log(m))$, we even require a modification to \cref{thm:introCycleFree} which shows that one can sample \emph{at smaller marginal rates} (instead of $1/2$) while still ensuring unique cycle survival. We then progressively alter these sampling rates until the minimum cycle length becomes $\approx \kappa \log(m)$.

Once the minimum cycle length becomes $\approx \kappa \log(m)$, instead of trying to enumerate short cycles in the graph, we instead \emph{directly invoke} \cref{thm:introCycleFree} to find \emph{large sets of cycle-free edges} (i.e., $\Omega(m)$ many edges without any cycles). Once we have such a set of edges, we then \emph{commit} to including these edges in our final spanning forest. Thus, from a graph on $n$ vertices, we now have collected $\Omega(m)$ of the edges we need for our ultimate spanning forest.

The key invariant now is that every $O(\log\log(m))$ rounds, we recover an $\Omega(1)$ fraction of the remaining edges we need for our spanning forest. Thus, after only $O(\log(m) \log\log(m))$ rounds, we recover a spanning forest of our starting graph.

We use the same strategy in the cographic setting. Namely, within each recursive iteration we invest $O(\log\log(m))$ rounds to eliminate small cuts, so that the minimum cut becomes $\approx \kappa \log(m)$ (to achieve this, we rely on a modification of \cref{thm:introConnected2} that allows sampling at smaller marginal rates). We then directly invoke \cref{thm:introConnected2} to find a large independent set and commit to including it in the basis, moving to the next iteration.

\subsection{Organization}

We start with preliminaries in \cref{sec:prelim}. In \cref{sec:connectedness} and \cref{sec:cycleFreeness}, we provide tight guarantees on how graph structure behaves under bounded independence sampling: \cref{sec:connectedness} shows that graphs with logarithmic min-cut remain connected under bounded-independence sampling, while \cref{sec:cycleFreeness} shows that graphs with large minimum cycle length become cycle-free. 

In \cref{sec:connectedness-from-dnfs}, we give a construction of a distribution of edges that has the same connectedness guarantees as $k$-wise independent sampling while needing only $O(\log m\cdot \log\log m)$ random bits.

In \cref{sec:framework}, present a general framework that underlies our basis finding algorithms for both graphic and cographic matroids. In \cref{sec:graphicDerand}, we use the cycle-freeness characterization to design near-optimal explicit, deterministic  parallel algorithms for finding bases of graphic matroids, and in \cref{sec:cographicDerand}, we use the aforementioned connectedness guarantee to provide explicit, deterministic parallel algorithms for finding bases of cographic matroids.

\section{Preliminaries}\label{sec:prelim}

We begin with providing tools and definitions that are necessary for the remainder of this paper.

An \emph{unweighted undirected graph} is defined as $G = (V, E)$, where $V$ is the set of vertices and $E$ is the set of edges. We always use $n, m$ to be the number of vertices and edges of a graph, respectively: $n = |V|, m = |E|$. A \emph{weighted} graph $G = (V, E, w)$ also has a function $w: E\to \mathbb{R}^+$, weighting the edges. For any subset of edges $A \subseteq E$, $w(A)$ denotes the sum of the weights of all edges in $A$. For any $v \in V$, $\delta(v)\subseteq V$ denotes the set of \emph{neighbors} of $v$ in $G$. For any subset of vertices $V' \subseteq V$, $G[V'] = (V', E', w')$ denotes the \emph{subgraph of $G$ induced by $V'$}, while $E(V')\subseteq E$ denotes the set of edges with both endpoints in $V'$. We say that a set of edges $C \subseteq E$ is a \emph{cut} in $G$ if and only if there is a set $S \subseteq V$ such that $C$ is the set of edges between $S$ and $\bar{S}$.

For any subset of vertices $S\subseteq V$, $G/S$ denotes the result of \emph{contracting} $G$ on a set $S$, i.e. identifying all the vertices of $S$. Specifically, $G/S = ((V\setminus S)\cup\{s^*\}, E')$, where $E' = \{(u, v): (u, v)\in E\text{ and } u, v \not \in S\}\cup \{(u, s^*): (u, s)\in E\text{ for some } s \in S\}$. Note that $G/S$ is a multigraph: if $u$ is connected to multiple edges in $S$, we keep all these edges as parallel edges in $G/S$.

Importantly, contracting a graph on a subset of vertices $S\subseteq V$ does not decrease its minimum cut, since it only limits the set of all cuts to those that do not cut through $S$.

We say that $f: \mathbb{N}\to \mathbb{N}$ is $\poly(n)$ if there exist $n_0, c\in \mathbb{N}$ such that $\forall n \geq n_0$, $f(n) \leq n^c$.

Finally, we use $\tilde{O}$ notation to omit polylogarithmic factors: $\tilde{O}(f(n)) = O(f(n)\cdot \poly(\log n))$.

\subsection*{Bounded-Independence Sampling}

We start by recalling the notion of a $\delta$-almost $k$-wise independent distribution:

\begin{definition}
	We say that a distribution $X = (x_1, \dots x_n)$ is \emph{$\delta$-almost $k$-wise independent} with marginals $p$ if for all subsets $S \subseteq [n], |S| \leq k$, we have
	\[
	d_{\mathrm{TV}}(U(S), X(S)) \leq \delta,
	\]
    where $d_{\mathrm{TV}}$ is the \emph{total variation distance} and $U(S)$ denotes the uniform distribution (with marginals $p$) over the set $S$. When $\delta = 0$, i.e. $X(S) = U(S)$ $\forall S: |S|\leq k$, we say that $X$ is $k$-wise independent.
\end{definition}

The work of Naor and Naor \cite{NN90} explicitly constructed such $\delta$-almost $k$-wise independent random variables of small size:

\begin{theorem}[\cite{NN90}]\label{thm:boundedIndSampling}
	There is an explicit construction of a $\delta$-almost $k$-wise independent distribution $X = (x_1, \dots x_n)$ with marginals $1/2$ and seed length $O(k + \log\log(n) + \log(1 / \delta))$.
\end{theorem}

By ``explicit'' construction of a distribution $X$ on $\{0, 1\}^n$, we mean that there is a polynomial time algorithm $G:\{0, 1\}^r\to \{0, 1\}^n$ (where $r$ is the \emph{seed length} or the \emph{number of random bits} used to generate $X$), such that $X = G(U_r)$.

When $\delta = 0$ (exact $k$-wise independence) we have the following space bounds:

\begin{theorem}[\cite{vadhan2012pseudorandomness}]\label{thm:boundedIndSamplingKWise}
	There is an explicit construction of a $k$-wise independent distribution $X = (x_1, \dots x_n)$ with marginals $1/2$ and seed length $O(k\log(n))$.
\end{theorem}

The following reduction allows us to construct bounded-independence distributions with arbitrarily small marginal probabilities:

\begin{algorithm}[H]
\SetKwInOut{Input}{input}
\SetKwInOut{Output}{output}
\caption{SampleLowerMarginals}\label{alg:arbMarginals}
\Input{\begin{minipage}[t]{0.85\linewidth}$n, k \in \mathbb{N}, \delta \geq 0, p$ s.t. $\log(1/p)\in\mathbb{Z}$; \newline $(y_1, \dots y_{n \log(1/p)})$ sampled from a $\delta$-almost $k \log(1/p)$-wise independent distribution $Y$ with marginals $1/2$.
\end{minipage}}

\Output{$X = (x_1, ..., x_n)$, where $x_i = \prod_{j = (i-1)\log(1/p)+1}^{i \log(1/p)} y_j$ for every $i\in [n]$.}
\end{algorithm}

\begin{lemma}\label{cor:arbMarginals}
	For any $n, k \in \mathbb{N}$, $\delta \geq 0$, $p$ such that $\log(1/p)$ is an integer, and $\delta$-almost $k \log(1/p)$-wise independent distribution $Y = (y_1, \dots y_{n \log(1/p)})$ with marginals $1/2$, the distribution $X = (x_1, \dots x_n)$ constructed using \cref{alg:arbMarginals} is $\delta$-almost $k$-wise independent with marginals $p$.
\end{lemma}

\begin{proof}
	We let $f:\{0,1\}^{n \log(1/p)} \rightarrow \{0,1\}^{n}$ denote the transformation $Y\mapsto X$, where $X, Y$ are defined as in \cref{alg:arbMarginals}, and we let $T_i = [(i-1)\log(1/p)+1, i \log(1/p)]$ denote these indices that $x_i$ depends on. Observe that for this transformation we have $\Pr[x_i = 1] = \Pr[\prod_{j \in T_i} y_j = 1] = (1/2)^{\log(1/p)} = p$. 
	
To see the $k$-wise independence condition, consider the uniform distribution $U'$ over $n \log(1/p)$ variables with marginals $1/2$. Under this distribution, if we perform the same grouping of the variables into groups of size $\log(1/p)$, we immediately get the independent distribution $U$ of $n$ Bernoulli trials with marginals $p$. For a set $S \subseteq [n]$, we let $S' = \bigcup_{i \in S} T_i\subseteq [n\log(1/p)]$, and importantly, $|S'| = |S| \cdot \log(1/p)$, as each $T_i$ is of size $\log(1/p)$. To conclude, we then have that 
	\[
	d_{\mathrm{TV}}(U\vert_S, X\vert_S) = d_{\mathrm{TV}}(f(U')\vert_S, f(Y)\vert_S) \leq d_{\mathrm{TV}}(U'(S'), Y(S')) \leq \delta,
	\]
	as $|S'| \leq k \log(1/p)$, and $Y$ is a $\delta$-almost $k \log(1/p)$-wise independent distribution.
\end{proof}

The support size of the distribution $X$ obtained from \cref{alg:arbMarginals} is the same as the support size of the original distribution $Y$ with marginals $1/2$. Hence, for $\delta > 0$, using \cref{thm:boundedIndSampling} with \cref{alg:arbMarginals} we get a $\delta$-almost $k$-wise independent distribution with marginals $p$ and seed length $O(k\log(1/p) + \log\log(n\log(1/p)) + \log(1 / \delta))$. Similarly, for $\delta = 0$, using \cref{thm:boundedIndSamplingKWise} with \cref{alg:arbMarginals} we get a $k$-wise independent distribution with marginals $p$ and seed length $O(k\log(1/p)\log(n\log(1/p)))$.

\bigskip

\begin{definition}
A distribution $X = (x_1, ..., x_n)$ on $\{0, 1\}^n$ is \emph{$k$-wise $\epsilon$-biased} if for all subsets $S \subseteq [n]$, $|S|\leq k$
\[
\left|\Pr\left[\bigoplus_{i\in S}x_i = 1\right] - \Pr\left[\bigoplus_{i\in S}x_i = 0\right]\right| \leq \epsilon
\]

If $k = n$, we say that $X$ is \emph{$\epsilon$-biased}.
\end{definition}

Naor and Naor gave an explicit construction of $k$-wise $\epsilon$-biased random variables (which in fact implies \cref{thm:boundedIndSampling} by setting $\epsilon = \delta/2^{k/2}$):

\begin{theorem}[\cite{NN90}]\label{thm:low-bias-distribution}
There is an explicit construction of a $k$-wise $\epsilon$-biased distribution $X = (x_1, \dots x_n)$ that can be generated using $O(\log k + \log\log n + \log(1/\epsilon))$ random bits.
\end{theorem}

\section{Connectedness Under Bounded-Independence Edge Subsampling}\label{sec:connectedness}

\subsection{Electric Flow and Effective Resistance}
\begin{definition}[Laplacian Matrix]
Given an undirected weighted graph $G = (V, E, w)$ on $n$ vertices, the \emph{Laplacian matrix} $\Lap_G\in \mathbb{R}^{n\times n}$ of $G$ is defined as:
\[
\Lap_G = \D_G - \A_G,
\]
where $\D_G\in \mathbb{R}^{n\times n}$ is the diagonal matrix of weighted degrees of $G$, and $\A_G$ is the weighted adjacency matrix of $G$: $(\A_G)_{ij} = (\A_G)_{ji} = w_{ij}$.
\end{definition}

It is easy to verify that for any graph $G$, $\Lap_G$ is positive semidefinite. Indeed, for any $\vx \in \mathbb{R}^n$, the Laplacian quadratic form is: $\vx^T\Lap_G\vx = \sum_{(u, v)\in E}w(u, v)\cdot (\vx(u) - \vx(v))^2 \geq 0$ since the weights are non-negative.

Since $\Lap_G$ is symmetric, it has a basis of eigenvectors. Hence, we use $\lambda_1, \lambda_2, ..., \lambda_n$ to denote the eigenvalues of $\Lap_G$ in increasing order ($\lambda_1 \leq \lambda_2 \leq ...\leq \lambda_n$). Since $\Lap_G\vec{1}^n = 0$, $\lambda_1 = 0$ for any graph $G$. It can further be verified that $\lambda_2 = 0$ iff $G$ is not connected.

\begin{definition}[Spectral Approximation]
For $\epsilon \geq 0$, a graph $H$ is an \emph{$\epsilon$-spectral approximation} of a graph $G$ if for all $\vx \in \mathbb{R}^n$ we have that $\vx^T\Lap_H \vx \in (1\pm \epsilon)\cdot \vx^T\Lap_G\vx$. In this case we write $\Lap_H \cong_\epsilon \Lap_G$.
\end{definition}

\begin{definition}[Network Flow]\label{defn:flow}
Given a graph $G = (V, E, w)$, some external flows $\iext\in \mathbb{R}^V$, and an arbitrary orientation of edges $E^\pm$, a \emph{flow} is a function $\vf\in \R^{E^\pm}$ such that $\vf(u, v) = -\vf(v, u)$ $\forall (u, v)\in E^\pm$ and
for any $v \in V$: 
\[
    (\iext)_v + \sum_{u \in \delta(v)} \vf(u, v) = 0\tag*{(flow conservation)}
\]
\end{definition}
The flow conservation law simply implies that no flow disappears or is created at any vertex. In contrast to arbitrary network flows, an electric flow also has to obey Ohm's law:
\begin{definition}[Electric Flow]
Given a graph $G = (V, E, w)$, external flows $\iext\in \mathbb{R}^V$, and an arbitrary orientation of edges $E^\pm$, a network flow $\vf\in \R^{E^\pm}$ (as in \cref{defn:flow}) is an \emph{electric flow} if there exists a potential function $\vp \in \mathbb{R}^V$ such that $\forall (u, v)\in E$:
\[
\vf(u, v) = w(u, v)(\vp(u) - \vp(v))\tag*{(Ohm's law)}
\]
\end{definition}

An important property that distinguishes an electric flow from just an arbitrary network flow is that it minimizes the electrical energy:

\begin{definition}[Electrical Energy]
Given a graph $G = (V, E, w)$ and some network flow $\vf\in \R^{E^\pm}$ (as in \cref{defn:flow}) in this graph, the \emph{electrical energy} of $\vf$ is defined as:
\[
\mathcal{E}(\vf) = \sum_{e\in E^\pm}\frac{\vf(e)^2}{w(e)}
\]
\end{definition}

\begin{fact}
Given a graph $G = (V, E, w)$ and external currents $\iext\in \mathbb{R}^V$, an electric flow $\vf\in \R^{E^\pm}$ is the unique network flow (with the given external currents) minimizing $\mathcal{E}(\vf)$.
\end{fact}

\begin{definition}[Effective Resistance]
Given a weighted undirected graph $G = (V, E, w)$, the \emph{effective resistance} between $u, v \in V$ is the potential difference between $u$ and $v$ induced by the unit $u$-$v$ electric flow (i.e. the unique electric flow induced by setting $\iext(u) = 1$, $\iext(v) = -1$, and $\iext(s) = 0$ for $s \not\in \{u, v\}$):
\[
\operatorname{Reff}_G(u, v) = \vp(u) - \vp(v)
\]
\end{definition}
\begin{fact}[Thomson’s Principle]
Given a graph $G = (V, E, w)$ and a unit $u$-$v$ electric flow $\vf_{uv} \in \mathbb{R}^{E^\pm}$, $\operatorname{Reff}_G(u, v) = \mathcal{E}(\vf_{uv})$.
\end{fact}
\begin{definition}[Effective Resistance Diameter]
\emph{Effective resistance diameter} of a graph $G = (V, E, w)$ equals the maximum effective resistance among all pairs of vertices:
\[
\operatorname{\operatorname{R_{diam}}}(G) = \max_{u, v \in V}\operatorname{Reff}_G
(u, v)\]
\end{definition}

\begin{definition}[Leverage Score]
Given a graph $G = (V, E, w)$, the \emph{leverage score} of an edge $(u, v) \in E$, denoted as $\ell(u, v)$, is the ratio between the effective resistance between its endpoints $\operatorname{Reff}_G(u, v)$ and the actual resistance of an edge $1/w(u, v)$:
\[
\ell(u, v) = w(u, v)\cdot \operatorname{Reff}_G(u, v)
\]
\end{definition}
\begin{fact}
Consider a connected graph $G = (V, E, w)$. If we sample a spanning tree of $G$ so that the probability of sampling a given tree is proportional to the product of the weights of edges in it, then the probability that a particular edge $(u, v)\in E$ is included in the sampled tree equals $\ell(u, v)$. 
\end{fact}

\subsection{Graph Clustering using Effective Resistance}
Alev et al. \cite{DBLP:conf/innovations/AlevALG18} show that any graph can be partitioned into components with small effective resistance diameter:
\begin{theorem}[\cite{DBLP:conf/innovations/AlevALG18}]\label{thm:graph-clustering-eff-resistance}
Given a weighted graph $G = (V, E, w)$, and a large enough parameter $\delta > 1$, there is an algorithm with time complexity $\tilde{O}\left(m \cdot n \cdot \log \left(\frac{w(E)}{\min _e w(e)}\right)\right)$ that finds a partition $V = \bigcup_{i=1}^h V_i$ satisfying
\begin{enumerate}
    \item $w\left(E-\cup_{i=1}^h E\left(V_i\right)\right) = O\left(w(E)/\delta\right)$,
    \item $\operatorname{R_{diam}}(G\left[V_i\right]) = O\left(\delta^3 \cdot \frac{|V|}{w(E)}\right)$ for all $i = 1, ..., h$.
\end{enumerate}
\end{theorem}

Making a constant $\delta$ large enough gives us the following corollary:

\begin{corollary}\label{cor:graph-clustering-eff-resistance}
There exists an absolute constant $\alpha$ such that given a weighted graph $G = (V, E, w)$, there is an algorithm with time complexity $\tilde{O}\left(m \cdot n \cdot \log \left(\frac{w(E)}{\min _e w(e)}\right)\right)$ that finds a partition $V = \bigcup_{i=1}^h V_i$ satisfying:
\begin{enumerate}
    \item $w\left(E-\cup_{i=1}^h E\left(V_i\right)\right) \leq w(E)/2$
    \item $\operatorname{R_{diam}}(G\left[V_i\right])\leq \alpha \frac{|V|}{w(E)}$ for all $i = 1, ..., h$.
\end{enumerate}
\end{corollary}

\subsection{Spectral Sparsification via Bounded Independence Sampling}

Doron et al. \cite{Doron2020SpectralSV} give an algorithm for spectral sparsification via bounded independence sampling:

\begin{algorithm}
\caption{$\operatorname{Sparsify}(G = (V, E, w), \{\tilde{R}_{ab}\}_{(a, b)\in E}, k, \epsilon, \delta)$}
\begin{enumerate}
\item Initialize $H$ to be the empty graph on $n=|V(G)|$ vertices.
\item Set $s \leftarrow \frac{18 e \log n}{\varepsilon^2} \cdot\left(\frac{n}{\delta}\right)^{2 / k}$.
\item For every edge $(a, b) \in E$, set $p_{a b} \leftarrow \min \left\{1, w_{a b} \cdot \widetilde{R}_{a b} \cdot s\right\}$
\item For every edge $(a, b) \in E$, add $(a, b)$ to $H$ with weight $w_{a b} / p_{a b}$ with probability $p_{a b}$. Do this sampling in a $k$-wise independent manner.
\item Return $H$.
\end{enumerate}
\end{algorithm}

\begin{theorem}[\cite{Doron2020SpectralSV} Theorem 3.1]\label{thm:spectral-sparsification}
Let $G = (V, E, w)$ be an undirected connected weighted graph on $n$ vertices with Laplacian $\Lap_G$ and effective resistances $R = \{R_{ab}\}_{(a, b)\in E}$. Let $0 < \epsilon < 1$, $0 < \delta < 1/2$, and let $k \leq \log n$ be an even integer. Let $H$ be an output of $\operatorname{Sparsify}(G, R, k, \epsilon, \delta)$ and let $\Lap_H$ be its Laplacian. Then, with probability at least $1 - 2\delta$ we have:
\begin{enumerate}
    \item $\Lap_H \cong_\epsilon \Lap_G$,
    \item $H$ has $O\left(\frac{1}{\delta^{1+2 / k}} \cdot \frac{\log n}{e^2} \cdot n^{1+\frac{2}{k}}\right)$ edges.
\end{enumerate}
\end{theorem}

Note that if $G$ is connected and $\Lap_H \cong_\epsilon \Lap_G$ for some $\epsilon\geq 0$, then $H$ must also be connected. Thus, setting $\epsilon = 1/2$, $k = \Theta(\log n)$, and $\delta = 1/\text{poly}(n)$, get that $s = \Theta(\log n)$ and $p_{ab} = \Theta(w_{ab}\cdot R_{ab}\cdot \log n)$, giving us the following corollary:
\begin{corollary}\label{cor:spectral-sparsification}
Let $G = (V, E, w)$ be an undirected connected weighted graph on $n$ vertices with effective resistances $\{R_{ab}\}_{(a, b)\in E}$. Then, for any $\delta = 1/\poly(n)$, sampling edges of $G$ in a $\Theta(\log n)$-wise independent manner with marginals $\Theta(w_{ab}\cdot R_{ab}\cdot \log n)$ for each $(a, b)\in E$ yields a connected graph $H$ with probability $1 - 2\delta$.
\end{corollary}

\subsection{Leverage Scores and Minimum Cut Duality}\label{sec:leverageScoreSec}

In this section, we prove the following theorem:
\begin{theorem}\label{thm:lev-scores-min-cut}
Given an unweighted graph $G = (V, E)$ with $m$ edges and minimum cut $c$, there exists a weighting $w: E\to \R$ such that the weighted graph $G' = (V, E, w)$ satisfies the following:
\begin{itemize}
    \item for each $e \in E$, the leverage score is $\ell(e) \leq 4\alpha/c$, where $\alpha$ is an absolute constant from \cref{cor:graph-clustering-eff-resistance},
    \item $w_{\max}/w_{\min} = m^{O(\log m)}$.
\end{itemize}
Moreover, the converse is also true: if there exists a weighting such that all leverage scores are at most $1/c$, the minimal cut of the original graph is $\geq c$.
\end{theorem}

We prove the theorem by recursively decomposing our graph into components with high effective resistance diameter, reweighting the edges within each component, and then contracting the graph on these components. We start by establishing the following useful lemma:

\begin{lemma}\label{lemma:energy-eff-resistance-diam}
Given a graph $G = (V, E, w)$ of effective resistance diameter $R$, and two disjoint vertex sets $s_1, s_2, ..., s_k$ and $t_1, t_2, ..., t_\ell$, let $\vf$ be the unique electric flow induced by injecting current $\alpha_i$ to $s_i$ and extracting $\beta_j$ from $t_j$ so that $\sum_i\alpha_i = \sum_j\beta_j = 1$. Then $\mathcal{E}(\vf) \leq R$.
\end{lemma}
\begin{proof}
Let $\vf_{ij}$ be a unit electric flow from $s_i$ to $t_j$. Then:
\[
\mathcal{E}(\vf_{ij}) = \operatorname{Reff}(s_i, t_j) \leq R
\]

Now, consider the distribution of pairs $(I, J)$ such that $P[I = i, J=j] = \alpha_i\beta_j$. If we pick the source $s_I$ and the sink $s_J$ according to this distribution, and induce a current as above, we will get some flow $\vf_{IJ}$, which in expectation will be:
\[
\mathbb{E}\left[\vf_{IJ}\right] = \sum_{i,j}\alpha_i\beta_j\vf_{ij} = \tilde{\vf}
\]
Importantly, $\tilde{\vf}$ is not necessarily an electric flow, since it might not obey Ohm's law. However, the energy of $\tilde{\vf}$ bounds the energy of the true electric flow with the same external currents.

Finally, note that by definition, $\mathcal{E}(\vf) = \sum_e\vf(e)^2/w_e = \vf^T\vec{R}\vf$, where $\vec{R}$ is the diagonal matrix of resistances. Since $\vec{R}$ is positive semi-definite, then $\mathcal{E}: \mathbb{R}^n\to \mathbb{R}, \vf\mapsto \mathcal{E}(\vf)$ is convex, meaning that by Jensen's inequality:
\[
\mathcal{E}(\tilde{\vf}) \leq \mathbb{E}_{(I,J)}\left[\mathcal{E}(\vf_{IJ})\right] = \sum_{i,j}\alpha_i\beta_j\mathcal{E}(\vf_{ij}) \leq R\sum_{i,j}\alpha_i\beta_j = R
\]
Moreover, note that the external currents of $\tilde{\vf}$ match exactly those of $\vf$: the current injected at $s_i$ is $\sum_j\alpha_i\beta_j = \alpha_i$ and the current extracted from $t_j$ is $\sum_i \alpha_i\beta_j = \beta_j$. Since the actual electrical flow minimizes the energy, then the energy of $\vf$ is $\mathcal{E}(\vf) \leq \mathcal{E}(\tilde{\vf}) \leq R$.
\end{proof}

We use this lemma to show that when a graph is contracted on some subgraphs, the effective resistance diameter of the contracted graph can be combined with effective resistance diameters of contracted components to bound the diameter of the original graph.

\begin{lemma}\label{lemma:contraction-eff-resistance-bound}
Given a graph $G = (V, E, w)$ and its partitioning $V = \sqcup_{i = 1}^h V_i$, let the contracted graph $G' = (V', E', w')$ be constructed by identifying vertices in the same partition set. Assume $\operatorname{R_{diam}}(G') \leq R_1$ and $\operatorname{R_{diam}}(G[V_i]) \leq R_0$ for all $i$. Then $\operatorname{R_{diam}}(G) \leq R_1 + hR_0$.
\end{lemma}
\begin{proof}
Consider any $s, t \in V$. We will show that $\operatorname{Reff}(s, t) \leq R_1 + hR_0$. Note that if $s, t \in V_i$ for some $i$, then the effective resistance between $s$ and $t$ in $G$ is bounded by the effective resistance in $G[V_i]$, meaning that $\operatorname{Reff}(s, t) \leq R_0 \leq R_1 + hR_0$. Hence, assume $s \in V_i$, $t\in V_j$ for some $i \neq j$.

Note that the effective resistance equals the minimal energy of a unit flow passing from $s$ to $t$. Hence, it is sufficient to construct some unit flow (not necessarily an electric flow) $\vf$ such that $\mathcal{E}(\vf) \leq R_1 + hR_0$.

Let $s', t' \in V'$ be the vertices of $G'$ that correspond to contracted components $V_i, V_j$ in $G$. Consider an electric unit flow $\vf'$ from $s'$ to $t'$ in $G'$: since $\operatorname{Reff}_{G'}(s', t') \leq \operatorname{R_{diam}}(G') \leq R_1$, then $\mathcal{E}_{G'}(\vf') = \sum_{e\in E'}\vf'(e)^2/w'(e) \leq R_1$. It remains to lift this flow to a flow in $G$.

Our lifted flow $\vf$ would match $\vf'$ on the flows of all edges crossing between partitions: namely, for each edge $e'\in E'$ of $G'$, if the corresponding edge in $G$ is $e\in E$, then we would force $\vf(e) = \vf'(e')$. Then, consider any subgraph $G[V_i]$ of $G$. The currents across the edges going between partitions force particular external flows to $G[V_i]$: specifically, if we limit our view to $G[V_i]$, any vertex $v\in V_i$ has external flow equaling the sum of all flows going to $v$ through these crossing edges (with the only exception of vertices $s$ and $t$, where the external currents are $1$ and $-1$, respectively). These external currents force a unique electric flow $\vf_i$ through $G[V_i]$. Since the total external current flowing in $G[V_i]$ equals the total current flowing out of $G[V_i]$ and is bounded by $1$ (the total current injected in $G$), then by \cref{lemma:energy-eff-resistance-diam} we get that the energy of this flow is $\mathcal{E}_{G[V_i]}(\vf_i) \leq \operatorname{R_{diam}}(G[V_i]) \leq R_0$.

Hence, we get the flow $\vf'$ fixing the flow across the edges going in between partitions, and in each subgraph $G[V_i]$ we get a flow $\vf_i$ consistent with the currents external to $G[V_i]$. Combining these flows, we get a unit $s$-$t$ flow $\vf$ in the original graph $G$. Finally, since the energy of the flow only depends on the flow across each edge and edge weights, we can safely combine the eneries of each of these flows to get the energy of $\vf$:
\[
\mathcal{E}(\vf) = \mathcal{E}(\vf') + \sum_{i = 1}^h\mathcal{E}(\vf_i) \leq R_1 + hR_0
\]

Finally, since the effective resistance between $s$ and $t$ in $G$ equals the minimal energy of a unit flow from $s$ to $t$, the above expression bounds it from above: $\operatorname{Reff}_G(s, t) \leq \mathcal{E}(\vf) \leq R_1 + hR_0$.
\end{proof}

With this lemma, we are finally ready to prove \cref{thm:lev-scores-min-cut}:

\begin{proof}[Proof of \cref{thm:lev-scores-min-cut}]

The converse is fairly straightforward to see. Consider the weighting for which the leverage scores are bounded by $1/c$ and take any cut of $G$ that has $k$ edges. Then sample a spanning tree of $G'$ with the probabilities proportional to the product of the edge weights in the tree. It is known that the probability of any edge being included in the spanning tree equals its leverage score, meaning that each cut edge will be included with probability at most $1/c$. Then, by the union bound, the probability that at least one cut edge is included is at most $k/c$. But one edge of any cut is definitely included in any spanning tree, meaning that $k/c \geq 1$, thus giving that any cut has size $k$ is $\geq c$. 

Now let us prove the main direction. Let $\Delta\in \R$ be the parameter to be chosen later. We will reweight the graph using \cref{alg:edge-weighting}.
\begin{algorithm}[h]
\caption{}\label{alg:edge-weighting}
\begin{enumerate}
    \item Set $i = 0$ and $G_0 = G$
    \item\label{alg:edge-weighting-step2} On recursive step $i$, find a partition of $G_i = (V_i, E_i)$ into subgraphs $V_i = \sqcup_{j=1}^hV_{ij}$ according to \cref{cor:graph-clustering-eff-resistance}, meaning that the number of edges going between partition sets is $\leq |E_i|/2$, and effective resistance diameter of $G_i[V_{ij}]$ is $\leq \alpha\cdot |V_i|/|E_i|$ for all $j$ (where $\alpha$ is a constant from \cref{cor:graph-clustering-eff-resistance}).
    \item For each edge $e$ in each $G_i[V_{ij}]$, set $w(e) = 1/\Delta^i$.
    \item Contract each $V_i$ in a single vertex, obtaining a multigraph $G_{i+1} = (V_{i+1}, E_{i+1})$, where $E_{i+1} = \{(u, v) \in E_i: u \in V_{ij}, v\in V_{ik} \text{ for } j\neq k\}$
    \item If $|V_{i+1}| > 1$, return to step 2 for $i+1$, otherwise output $w$.
\end{enumerate}
\end{algorithm}

Consider any step $i$ of the algorithm and any subgraph $G_i[V_{ij}]$ of $G_i$ that would be contracted on this step. By a slight abuse of notation, write $G[V_{ij}]$ to be the subgraph of the original graph $G$ induced by all the vertices in $V$ that are contracted to vertices $V_{ij}$ on steps $0$ to $i-1$.

We will prove by induction that on each step $i$, for each subgraph $G[V_{ij}]$ of the original graph, the effective resistance diameter (with respects to the weights $w$ assigned) is:
\[
\operatorname{R_{diam}}\left({G[V_{ij}]}\right) \leq 2\alpha\cdot \Delta^i\cdot \left(1 + \frac{n}{\Delta}\right)^i\cdot \frac{1}{c},
\]
where $\alpha$ is a constant from \cref{cor:graph-clustering-eff-resistance}.

Note that since the minimal cut of $G$ is $c$, then each vertex has degree at least $c$, meaning that the total number of vertices in $G$ is at least $c|V|/2$. Then $|V|/|E| \leq |V|/(c|V|/2) = 2/c$, so by our choice of partitioning on step \ref{alg:edge-weighting-step2}:
\[
\operatorname{R_{diam}}\left({G[V_{0j}]}\right) = \operatorname{R_{diam}}\left({G_0[V_{0j}]}\right) \leq \alpha\frac{|V|}{|E|} = 2\alpha\cdot \frac{1}{c},
\]
proving the base case of induction.

Then, consider any step $i \geq 1$ and any subgraph $G_i[V_{ij}]$ of $G_i$. Notice that the minimal cut of $G_i$ is at least $c$, since contracting on subgraphs only limits the set of all cuts to those which don't cut through contracted subgraphs, meaning that it does not decrease the min cut. Hence, $|V_i|/|E_i| \leq 2/c$, so the effective resistance diameter of a subgraph found on Step \ref{alg:edge-weighting-step2} of \cref{alg:edge-weighting} is $\leq 2\alpha/c$. However, note that when we set the weights of all edges of $G_{i}[V_{ij}]$ to $1/\Delta^i$, the effective resistance diameter given above (which is the diameter for unit weights) will be scaled by $\Delta^i$ (since the resistance of each edge is scaled by $\Delta^i$), becoming:
\[
\operatorname{R_{diam}}(G_i[V_{ij}]) \leq 2\alpha\cdot  \Delta^i\cdot \frac{1}{c}
\]

Then, consider any vertex $v_k \in V_{ij}\subseteq V_i$. This vertex was created by contracting some component $V_{(i-1)k}$ on the previous recursive step of the algorithm, so by the inductive hypothesis we get that:
\[
\operatorname{R_{diam}}\left({G[V_{(i-1)k}]}\right) \leq 2\alpha\cdot  \Delta^{i-1}\cdot \left(1 + \frac{n}{\Delta}\right)^{i-1}\cdot \frac{1}{c}
\]

Note that $G_i[V_{ij}]$ of effective resistance diameter $2\alpha \Delta^i/c$ is obtained from $G[V_{ij}]$ by contracting at most $n$ such subgraphs $G[V_{(i-1)k}]$ of effective resistance diameter $2\alpha\cdot \Delta^{i-1}\cdot \left(1 + \frac{n}{\Delta}\right)^{i-1}\cdot \frac{1}{c}$ each, meaning that by applying \cref{lemma:contraction-eff-resistance-bound} we can bound $\operatorname{R_{diam}}(G[V_{ij}])$ as:
\[
\operatorname{R_{diam}}(G[V_{ij}]) \leq 2\alpha\Delta^i\cdot \frac{1}{c} + 2\alpha n\cdot \Delta^{i-1}\cdot\left(1 + \frac{n}{\Delta}\right)^{i-1}\cdot \frac{1}{c} \leq 2\alpha\cdot \Delta^{i}\cdot\left(1 + \frac{n}{\Delta}\right)^{i}\cdot \frac{1}{c},
\]
which concludes our inductive proof for $\operatorname{R_{diam}}$.

Finally, consider any edge that was contracted on step $i$: $e\in G_i[V_{ij}]$. Then, clearly the effective resistance of $e$ in $G$ is upper bounded by the effective resistance of $e$ in a subgraph of $G$, $G[V_{ij}]$, which in turn is upper bounded by $\operatorname{R_{diam}}(G[V_{ij}])$. Then, since $w(e) = 1/\Delta^i$, we get that the leverage score of $e$ in $G$ is upper bounded by:
\[
\ell_G(e) = w(e)\operatorname{Reff}_G(e) \leq w(e)\operatorname{Reff}_{G[V_{ij}]}(e) \leq \frac{\operatorname{R_{diam}}(G[V_{ij}])}{\Delta^i} \leq 2\alpha\cdot \left(1 + \frac{n}{\Delta}\right)^{i}\cdot \frac{1}{c}
\]

To conclude the proof, we observe that on each recursive step $|E_i|$ is reduced by a factor of 2, meaning that the total number of steps the algorithm takes to complete is $\log m$. Then, picking $\Delta = \Theta(n\log m)$ is sufficient to ensure that $\left(1 + \frac{n}{\Delta}\right)^{i} \leq \left(1 + \frac{n}{\Delta}\right)^{\log m} \leq 2$, resulting in all leverage scores being $\leq 4\alpha/c$.

For such $\Delta$, get that $w_{\max}/w_{\min} = \Delta^{O(\log m)} = m^{O(\log m)}$.

\end{proof}

\subsection{Concluding Connectedness Under Subsampling}

Finally, we show that \cref{thm:lev-scores-min-cut} can be combined with \cref{thm:spectral-sparsification} to show that a graph with a large min cut stays connected even when edges are subsampled in a $k$-wise independent manner:
\begin{theorem}\label{thm:connectednessBoundedIndSampling}
Given an undirected unweighted graph $G = (V, E)$ on $m$ edges, if the minimum cut of $G$ is $c\log m$, then sampling the edges of $G$ in a $O(\log m)$-wise independent manner with marginals $O(1/c)$ yields a connected graph with probability $1 - 1 / \mathrm{poly}(m)$.
\end{theorem}
\begin{proof}
\cref{thm:lev-scores-min-cut} tells us that there is a weighting $w$ of $G$ such that for each $(a, b)\in E$, the leverage score is $w_{ab}\cdot \tilde{R}_{ab} \leq 4\alpha/c\log m$. Then, by \cref{cor:spectral-sparsification}, $O(\log m)$-wise independent sampling with marginals $p_{ab} = \Theta(w_{ab}\cdot \tilde{R}_{ab}\cdot \log n) = \Theta(1/c)$ gives a connected graph with probability $1 - 1 / \mathrm{poly}(m)$.
\end{proof}

We also have the following immediate corollary:

\begin{corollary}\label{cor:connectednessBoundedIndSampling}
Given an undirected unweighted graph $G = (V, E)$ on $m$ edges, there is a choice of constant $\kappa$ such that if the minimum cut of $G$ is $\kappa \log m$, then sampling the edges of $G$ in a $O(\log m)$-wise independent manner with marginals $1/2$ yields a connected graph with probability $1 - 1 / \mathrm{poly}(m)$.
\end{corollary}

In \cref{appendix:unique-survival}, we show that these results can be extended to show that by choosing the marginal probabilities correctly, we can ensure that there is a non-negligible probability that \emph{any} near-minimum cut is the \emph{unique} surviving cut in our sample.

\section{Connectedness Under Other Sampling Schemes}\label{sec:connectedness-from-dnfs}
\subsection{Fooling Read-Once DNF Formulas}

\begin{definition}
A \emph{DNF formula} $\phi$ is of the form $\phi = \bigvee_{i=1}^mC_i$, where each term $C_i$ is an AND of literals (variables or negations). A DNF is said to be \emph{read-once} if every variable appears in at most one term. 
\end{definition}

\begin{definition}
We say that the distribution $X = (x_1, \dots x_n)$ \emph{$\delta$-fools} a real function $f:\{0, 1\}^n\to \mathbb{R}$ if 
\[
\left|\mathbb{E}[f(X)] - \mathbb{E}[f(U_n)]\right| \leq \delta
\]
\end{definition}

\begin{observation}[\cite{DETT10}]\label{obs:fooling-and}
If $\phi:\{0, 1\}^n\to \{0, 1\}$ is an AND of some subset of literals, then $\phi$ is $\epsilon$-fooled by every $\epsilon$-biased probability distribution.
\end{observation}

De et al. \cite{DETT10} showed that low bias distributions fool read-once DNFs:

\begin{theorem}[\cite{DETT10}]\label{thm:fooling-dnfs}
Let $\phi$ be a read-once DNF formula with $m$ terms. For every $1 \leq \ell \leq m$, $\epsilon$-biased distributions $O(2^{-\Omega(\ell)}+\epsilon m^\ell)$-fool $\phi$. In particular, we can $\delta$-fool $\phi$ by an $\epsilon$-biased distribution for $\epsilon = m^{-O(\log(1/\delta))}$.
\end{theorem}

\subsection{Connectedness Given Large Minimum Cut}

In this section, we show a near-optimal sampling scheme that preserves connectedness while needing only $O(\log m\cdot \log\log m)$ random bits.

\begin{theorem}\label{thm:connectedness-from-dnfs}
There exists an absolute constant $c$ such that the following holds. Let $X^{(1)}, \dots,X^{(c\log n)}$ be independent copies of any distribution on $\{0, 1\}^{m\times \log\ell}$ that $0.1$-fools the following collection of read-once DNFs:
\[
\mathcal{F} = \left\{\bigvee_{i\in S}\left(\bigwedge_{j=1}^{\log \ell} x_{ij}\right): S\subseteq [m], |S|\leq \ell\right\}
\]
Then, for every undirected unweighted graph $G = (V, E)$ with $n$ vertices, $m$ edges, and minimum cut $\geq \ell$, subsampling the edges of $G$ using distribution $Y$ on $\{0, 1\}^m$ defined as
\[
Y_k = \bigvee_{i = 1}^{c\log n} \left(\bigwedge_{j=1}^{\log \ell} X_{kj}^{(i)}\right),
\]
yields a connected graph with probability $1-1/\poly(m)$.
\end{theorem}

\begin{proof}

Intuitively, the above construction is the following. We start by taking a distribution $X$ that fools read-once DNFs with $\widetilde{O}(\ell)$ variables. Grouping the elements of $X$, we obtain the distribution $\overline{X}$ on $\{0, 1\}^m$ with smaller marginals, where $\overline{X}_k=\bigwedge_{j=1}^{\log \ell}X_{kj}$. Then, in $O(\log n)$ independent rounds, we subsample the edges of $G$ using independent copies of $\overline{X}$, and then take the union of edges sampled in all rounds, yielding $Y = \bigvee_{i=1}^{c\log n}\overline{X^{(i)}}$ as claimed. 

Once a connected component is formed after some round of sampling, it does not matter for connectivity what edges inside this component are sampled in the subsequent rounds. Hence, for our analysis it would be useful to imagine that after each round of sampling we contract the graph on newly formed connected components.

Hence, consider any sampling round $i$ and any vertex $v$ of this round ($v$ is a result of contracting $G$ on the connected components of the subgraph of $G$ whose edge set indicator vector is $\overline{X^{(1)}}\lor \overline{X^{(2)}}\lor\dots\log \overline{X^{(i-1)}}$). We will bound the probability that we sample at least one edge incident to $v$ in this round. Since the minimum cut does not decrease after contracting, $v$ has degree at least $\ell$. Take any subset $S\subseteq [m]$ of edges incident to $v$, such that $|S| = \ell$. Then, one of the edges in $S$ is sampled in this round if and only if $\phi_v\left(X^{(i)}\right) = 1$, where
\[
\phi_v\left(x\right) = \bigvee_{k\in S}\overline{x}_k = \bigvee_{k\in S}\left(\bigwedge_{j=1}^{\log \ell} x_{kj}\right)
\]

Since $|S| = \ell$, $\phi_v\in \mathcal{F}$, and so it is $0.1$-fooled by $X^{(i)}$. Under the uniform distribution $U$ on $\{0, 1\}^{m\times \log \ell}$, each edge in $S$ is sampled by $\overline{U}$ with probability $(1/2)^{\log \ell} = 1/\ell$, and so 
\[
\Pr[\phi_v(U) = 1] = 1 - (1- 1/\ell)^{\ell} \geq 1-1/e\geq 0.6.
\]
Therefore, 
\[
\Pr\left[\phi_v\left(X^{(i)}\right) = 1\right] \geq 0.6 - 0.1 = 1/2.
\]

Therefore, in round $i$, each vertex $v$ finds a neighbor with probability at least $1/2$. Assume that in round $i$ (after some contractions) we have $n' \leq n$ vertices, and let $N$ be the number of vertices that do not find a neighbor in this round. Then, $\mathbb{E}[N] \leq n'/2$. Therefore, by Markov's inequality
\[
\Pr[N \geq 0.9n'] \leq \frac{\mathbb{E}[N]}{0.9n'} \leq \frac{n'/2}{0.9n'} < 0.6
\]
Therefore, with probability $\Omega(1)$, a constant fraction of vertices find a neighbor, meaning that the total number of vertices after contracting on newly created components is reduced by a constant factor. Since we only need $O(\log n)$ such reductions to contract everything to a single vertex, $X^{(i)}$'s are independent, and the success probability is constant within each round, Chernoff bound implies that with probability $1-1/\mathrm{poly}(m)$, after $O(\log m)$ steps we will contract all vertices. This implies that the union of all sampled edges would form a connected graph.

\end{proof}

Using the fact that low-bias distributions fool read-once DNFs (\cref{thm:fooling-dnfs}), we obtain the following result:
\begin{corollary}\label{cor:connectedness-distribution-dnfs}
There exists an absolute constant $\kappa$ and an explicit distribution $Y = (y_1, \dots y_m)$ with seed length $O(\log m\cdot \log \log m)$ and marginals $\leq 1/2$ such that for every undirected unweighted graph $G = (V, E)$ with $n$ vertices, $m$ edges, and minimum cut $\geq \ell = \kappa\log m$, subsampling the edges of $G$ under $X$ gives a connected graph with probability $1 - 1/\poly(m)$.
\end{corollary}

\begin{proof}
Let $X$ be a $\widetilde{O}(\log m)$-wise $(1/\mathrm{polylog}(m))$-biased distribution on $\{0, 1\}^{m\times \log\log m}$. By \cref{thm:low-bias-distribution}, we can construct such $X$ explicitly using $O(\log \log m)$ random bits. By \cref{thm:fooling-dnfs}, $X$ fools read-once DNFs with $\widetilde{O}(\log m)$ variables up to any desired constant probability, meaning that by \cref{thm:connectedness-from-dnfs}, sampling the edges of $G$ using $Y_k\in \{0, 1\}^m$,
\[
Y_k = \bigvee_{i = 1}^{c\log n} \left(\bigwedge_{j=1}^{\log \ell} X_{kj}^{(i)}\right),
\]
yields a connected graph with probability $1 - 1/\poly(m)$. To construct $Y$, we need $O(\log n)$ independent copies of $X$, meaning that the total number of bits needed to generate $Y$ is $O(\log m\cdot \log\log m)$.

It remains to bound the marginal probabilities. Consider any edge $e$. In any given round $i$, $e$ is sampled if and only if
\[
\bigwedge_{j=1}^{\log \ell} X_{ej}^{(i)} = 1
\]
Under uniform distribution, the above formula is satisfied with probability $(1/2)^{\log \ell}= 1/\ell = 1/\kappa\log m$. Since $\ell = \kappa\log m$, this formula is an AND of $O(\log\log m)$ variables, then by \cref{obs:fooling-and}, it is $(1/\mathrm{polylog}(m))$-fooled by $X^{(i)}$, meaning that we sample $e$ in this round with probability at most $2/\kappa \log m$. Hence, the probability that $e$ is not sampled in all $c\log(n)$ rounds is at least $(1 - 2/\kappa\log(m))^{c\log n} \geq 1/2$ for sufficiently large constant $\kappa$. Hence, the marginal probability for each edge is at most $1/2$.
\end{proof}

\subsection{Unique Cut Survival Under Subsampling}

The above approach easily generalizes to the setting of unique cut survival. 

\begin{theorem}\label{thm:unique-survival-from-dnfs}
There exists an explicit distribution $Y = (y_1, \dots y_m)$ with seed length $O(\log m\cdot \log\log m)$ such that for every undirected unweighted graph $G = (V, E)$ with $n$ vertices, $m$ edges, and minimum cut $\ell \leq \kappa\log n$ (for $\kappa$ as in \cref{cor:connectedness-distribution-dnfs}), and every cut $C\subseteq E$ of size $\in [\ell, 1.01\ell]$, subsampling the edges of $G$ using $Y$ yields a subset $Q\subseteq E$ such that with probability $> 0$
\begin{enumerate}
    \item $C\subseteq Q$;
    \item for every other cut $C'$, $C'\not\subseteq Q$.
\end{enumerate}
\end{theorem}
\begin{proof}
If $\ell \leq 100$, we take the uniform distribution on all subsets of $\leq 1.01\ell$ edges. This distribution has support size $\mathrm{poly}(n) = 2^{O(\log n)}$, and so can be generated using $O(\log m)$ random bits. Moreover, for each cut $C$ of size $\leq 1.01\ell$, there would be a sample containing exactly the edges of $C$. Now, assume $\ell \geq 100$.

For $\ell \geq 100$, we will construct the distribution $\widetilde{Y}$ on $\{0, 1\}^m$ such that for each cut $C$ of size $\leq 1.01\ell$, with probability $> 0$ \emph{no edge of $C$} is sampled under sampling from $\widetilde{Y}$, while for every other cut $C'$, \emph{at least one} edge is sampled. Then, by taking $Y_k = 1- \widetilde{Y}_k$, we immediately obtain the desired distribution $Y$.

We let $X^{(1)}, \dots,X^{(c\log n)}$ to be independent copies of a $\widetilde{O}(\log m)$-wise $(1/\mathrm{polylog}(m))$-biased distribution given by \cref{thm:low-bias-distribution}. We construct our distribution $\widetilde{Y}$ to be
\[
\widetilde{Y}_k = \bigvee_{i = 1}^{c\log n} \left(\bigwedge_{j=1}^{\log \ell} X_{kj}^{(i)}\right),
\]
where $c$ is some absolute constant to be fixed later. Note that by \cref{thm:low-bias-distribution}, each $X^{(i)}$ can be generated using $O(\log m)$ random bits, and so the total seed length of $\widetilde{Y}$ is $O(\log m\cdot \log \log m)$.

We again view this procedure as first grouping the elements of $X$ to obtain a distribution $\overline{X}_k = \bigwedge_{j=1}^{\log \ell}X_{kj}$ over $\{0, 1\}^m$, sampling the edges using $O(\log m)$ independent copies of $\overline{X}$, and then taking the union of edges sampled. In a given round $i$, no edge of $C$ is sampled if and only if $\phi_C\left(X^{(i)}\right) = 0$, where
\[
\phi_C\left(x\right) = \bigvee_{k\in C}\overline{x}_k = \bigvee_{k\in C}\left(\bigwedge_{j=1}^{\log \ell} x_{kj}\right),
\]
Note that $\phi_C$ is a read-once DNF with $\widetilde{O}(\log n)$ variables, meaning that by \cref{thm:fooling-dnfs}, it is $0.01$-fooled by $X^{(i)}$. Under the uniform distribution $U$ on $\{0, 1\}^{m\times \log\ell}$, each edge in $C$ is sampled with probability $(1/2)^{\log \ell} = 1/\ell$. Since $|C| \leq 1.01\ell$, 
\[
\Pr[\phi_C(U) = 0] = (1-1/\ell)^{|C|} \geq  (1-1/\ell)^{1.01\ell} \geq 0.99\cdot (1/e^{1.01}) > 0.36
\]
for $\ell \geq 100$. Therefore, under the distribution $X^{(i)}$, 
\[
\Pr\left[\phi_C\left(X^{(i)}\right) = 0\right] \geq 0.36-0.01 = 0.35.
\]
Therefore, since $X^{(1)}, \dots X^{(c\log n)}$ are independent, there is a non-zero probability that $\phi_C\left(X^{(i)}\right) = 0$ for all $i$, meaning that no edge of $C$ is sampled in $\widetilde{Y}$.

It remains to prove that conditioned on no edge of $C$ being sampled in $\widetilde{Y}$, we would sample at least one edge from each from each of the remaining cuts. The proof of this claim is essentially identical to the proof of \cref{thm:connectedness-from-dnfs}. We show that if we contract the connected components after each round of sampling, after $O(\log n)$ rounds we will get two components separated by $C$ with high probability (conditioned on not sampling edges in $C$). 

Consider any round $i$ and any vertex $v$ in this round. The set of edges incident to $v$, $C'$, is a cut $(\{v\}, V\setminus\{v\})$. We note that since the minimum cut is $\ell$, and $|C|\leq 1.01\ell$, we must have that $|C'\setminus C| \geq 0.2\ell$ (this is formally proven in \cref{clm:removeCutNewCutSize} of the Appendix). Hence, we can find a set $A\subseteq |C'\setminus C|$, $|A| = 0.2\ell$ of edges incident to $v$ that are not in $C$. We sample one of these edges in round $i$ if and only if $\phi_v\left(X^{(i)}\right) = 1$, where
\[
\phi_v(x) := \bigvee_{k\in A} \overline{x}_i = \bigvee_{k\in A}\left(\bigwedge_{j=1}^{\log \ell} x_{kj}\right).
\]
We need to bound $\Pr\left[\phi_v\left(X^{(i)}\right) = 1\mid \phi_C\left(X^{(i)}\right) = 0\right]$, i.e. the probability that $v$ finds a neighbor in round $i$ conditioned on no edge of $C$ being sampled. Note that:
\begin{equation*}
\begin{aligned}
\Pr\left[\phi_v\left(X^{(i)}\right) = 1\mid \phi_C\left(X^{(i)}\right) = 0\right] &= 1 - \Pr\left[\phi_v\left(X^{(i)}\right) = 0\mid \phi_C\left(X^{(i)}\right) = 0\right] \\ \\ &=1 - \frac{\Pr\left[\left(\phi_v\lor \phi_C\right)\left(X^{(i)}\right) = 0\right]}{\Pr\left[\phi_C\left(X^{(i)}\right) = 0\right]}
\end{aligned}
\end{equation*}
We have already seen that $\Pr\left[\phi_C\left(X^{(i)}\right) = 0\right] \geq 0.35$. Finally, note that $\phi_v\lor \phi_C$ is a read-once DNF with $\widetilde{O}(\log n)$ variables, meaning that it is $0.01$-fooled by $X^{(i)}$. Under the uniform distribution $U$ on $\{0, 1\}^{m\times \log \ell}$, 
\[
\Pr\left[(\phi_v\lor \phi_C)(U) = 0\right] = (1-1/\ell)^{|C| + |A|} \leq (1-1/\ell)^{\ell + 0.2\ell} \leq 1/e^{1.2} \leq 0.31,
\]
meaning that under $X^{(i)}$,
\[
\Pr\left[\left(\phi_v\lor \phi_C\right)\left(X^{(i)}\right) = 0\right] \leq 0.31 + 0.01 = 0.32
\]
. Therefore
\[
\Pr\left[\phi_v\left(X^{(i)}\right) = 1\mid \phi_C\left(X^{(i)}\right) = 0\right] \geq 1 - \frac{0.31}{0.35} > 0.1
\]
Therefore, conditioned on no edge of $C$ being sampled, in every round, each vertex still has finds a neighbor with probability $\geq 0.1$. Similarly to the proof of \cref{thm:connectedness-from-dnfs}, this implies that in each round, with probability $\Omega(1)$ a constant fraction of edges find a neighbor, meaning that the number of connected components decreases by a constant factor, and so in $O(\log n)$ rounds, we will get to two components separated by $C$ with high probability. Hence, taking $Y$ to be the complement of the subset sampled by $\widetilde{Y}$, we indeed get that for all cuts $C$ of size $\in [\ell, 1.01\ell]$, $C$ is the unique surviving cut in $Y$ with positive probability.

\end{proof}

\section{Cycle-Freeness Under Almost Bounded-Independence Sampling}\label{sec:cycleFreeness}

In this section, we consider using (almost) bounded-independence when sampling graphs with large girth. For these graphs, we show that even this bounded-independence sampling yields subsampled graphs with \emph{no cycles} with high probability. 

\subsection{Cycle-Freeness Under High Girth}

To start, we recall the work of Karp, Upfal, and Wigderson \cite{KUW85}, who showed the following claim that characterizes when cycles appear under \emph{uniform sampling} of edges:

\begin{claim}\label{clm:originalSampling}
	Let $G = (V, E)$ be a graph with $m$ edges and shortest cycle length $> 20 \log(m)$. A uniformly random subsample of $E$ at rate $1/2$ will (with high probability):
	\begin{enumerate}
		\item Contain at least $m/4$ edges.
		\item Not have any cycles. 
	\end{enumerate}
\end{claim}

We include the proof here for completeness:

\begin{proof}
	To start, observe that because the shortest cycle is of length $> 20 \log(m)$, this means that for any pair of vertices $i, j \in V$, there can be at most one path of length $10 \log(m)$ between $i, j$.
	
	Now, let us consider sampling $E$ at rate $1/2$ to yield $\hat{E}$. If $\hat{E}$ contains a cycle, then it must be a cycle of length $>20 \log(m)$, as $E$ itself had no cycles of shorter length, and $\hat{E} \subseteq E$. At the same time, if $\hat{E}$ contains a cycle $C$ of length $>20 \log(m)$, $C$ must contain some path of length $10 \log(m)$. Thus, in order to show that $\hat{E}$ does not contain any cycles, it suffices to show that $\hat{E}$ does not contain any paths of length $10 \log(m)$. To see this, we recall that there is at most one path of such length for each pair of $i,j \in V$, and thus there are at most 
	\[
	\binom{|V|}{2} \leq \binom{2|E|}{2} \leq 2m^2
	\]
	such paths. Because each path survives with probability $\leq \left ( \frac{1}{2} \right )^{10 \log(m)}$, we see that there is no such path with probability $\geq 1 - \frac{2m^2}{m^{10}} \geq 1 - \frac{2}{m^8}$. A simple Chernoff bound yields that the number of edges sampled is $\geq m/4$ with exponentially high probability, and thus conditions (1) and (2) are satisfied with probability $\geq 1- \frac{1}{m^7}$.
\end{proof}

Now, we show how to use $\delta$-almost $k$-wise independent distributions to derandomize the above sampling procedure:

\begin{claim}\label{clm:derandomizeSampling}
	Let $G = (V, E)$ be a graph with $m$ edges and shortest cycle length $> 20 \log(m)$. Consider sampling the edges of $G$ using a $(1/m^{100})$-almost $10 \log(m)$-wise independent distribution with marginals $1/2$. Then, there is a sample from the distribution which:
	\begin{enumerate}
		\item Contains at least $m/10$ edges.
		\item Does not have any cycles. 
	\end{enumerate}
\end{claim}

\begin{proof}
	As in the proof of \cref{clm:originalSampling}, it suffices to prove that none of the $\leq 2m^2$ paths of length $10 \log(m)$ survive the sampling. For each such path, let us denote the constituent edges by $S$. Because our sample space is $\delta$-almost $10\log(m)$-wise independent, we know that over a random sample $X = (x_1, \dots x_m)$ from our space,
	\[
	d_{\mathrm{TV}}(X(S), U(S)) \leq \frac{1}{m^{100}}.
	\]
	In particular, this means that 
	\[
	\Pr[X(S) \neq 1^S] \geq \Pr[U(S) \neq 1^S] - \frac{1}{m^{100}} = 1 - \frac{1}{m^{10}} - \frac{1}{m^{100}} \geq 1 - \frac{2}{m^{10}},
	\]
	which means that 
	\[
	\Pr[\text{path } S \text{ survives}] \leq \frac{2}{m^{10}}.
	\]
	Taking a union bound over all $\leq 2m^2$ paths then yields that with probability $\geq 1 - \frac{4}{m^8}$, no path of length $10 \log(m)$ survives, and thus that no cycle survives the sampling.
	
	Now, because we use a $\delta$-almost $10\log(m)$-wise independent distribution with marginals $1/2$, we know that $\E[\sum_{i = 1}^m X_i]  = m/2$. By a simple Markov bound, we know then that $\Pr[\sum_{i = 1}^m X_i \geq m/10] \geq 0.4$. So, by a union bound, we know that with probability $\geq 0.4 - \frac{4}{m^8}$, a sample from our distribution will have at least $m/10$ edges and not have any cycles. Because this probability is $>0$, this yields the stated claim. 
\end{proof}

We also can extend the result above to graphs with smaller shortest cycle length provided that we decrease the marginals accordingly:
\begin{corollary}\label{cor:smallerGirth}
    Given an undirected unweighted graph $G = (V, E)$ on $m$ edges, if the shortest cycle length of $G$ is $\ell \leq 20 \log(m)$, then sampling the edges of $G$ using \cref{alg:arbMarginals} with $p =(1/2)^{\left\lceil40 \log(m)/\ell\right\rceil}$ and $(1/m^{200})$-almost $\ell$-wise independence yields a connected graph with probability $1 - 1 / \mathrm{poly}(m)$.
\end{corollary}

\begin{proof}
    We construct an auxiliary graph $G'$, whereby we replace each edge in $G$ with a path of length $s = \left\lceil40 \log(m)/\ell\right\rceil$ (adding $s-1$ new vertices for each edge). Immediately, we can see that the shortest cycle in $G'$ is of length $\ell' \geq 40 \log(m)$. At the same time, the number of edges in $G'$ is $m' = sm < m^2$, since $s \leq m$ for sufficiently large $m$. Hence, we get that $\ell' \geq 40\log(m) > 20\log(m')$ and $\log(m') \leq 2\log(m)$, meaning that we can apply \cref{clm:derandomizeSampling} to $G'$. Namely, if one samples the edges of $G'$ using $(1/m^{2\cdot 100})$-almost $(2\cdot 10\log m)$-wise independent random bits with marginals $1/2$, the resulting graph is connected with probability $1 - 1 / \mathrm{poly}(m)$.

    Now, observe that our implementation of $\delta$-almost $\ell$-wise independent sampling with marginals $(1/2)^{\left\lceil 40 \log(m)/\ell\right\rceil}$ from \cref{alg:arbMarginals} exactly corresponds to the above approach.
\end{proof}

\subsection{Unique Cycle Survival Under Careful Sampling Rates}\label{sec:uniqueSurvival}

In this section, we extend the argument from the previous section to also address the setting of \emph{cycles uniquely surviving} under sampling:

\begin{claim}\label{clm:boundedLengthContract}
    Let $G = (V, E)$ be a graph with shortest cycle length $\ell$, and let $C \subseteq E$ denote the edges participating in a cycle of length $\in [\ell, 1.01 \ell]$. Now, let $G / C$ denote the result of contracting the graph $G$ on the edges $C$. Then, the shortest cycle length in $G / C$ is $\geq 0.2 \ell$.
\end{claim}

\begin{proof}
    Note a set of edges $C' \subseteq  G / C$ forms a cycle in $G / C$ if and only if $C' \cup C$ contains a cycle different from $C$ in $G$. Thus, we instead focus on cycles $\hat{C} \subseteq G$, $\hat{C}\neq C$, and lower bound $|\hat{C} \setminus C|$.
    
    It is easy to verify that for any pair of cycles $\hat{C}\neq C$, their symmetric difference $\hat{C}\oplus C = (\hat{C}\cup C)\setminus(\hat{C}\cap C)$ also contains a cycle (if we consider the edge-induced subgraph, any vertex in the symmetric difference would have an even degree).
    
    Assume $|\hat{C}\setminus C| < 0.2\ell$. Then, since $|\hat{C}| \geq \ell$ ($\ell$ is the minimum cycle length), we have that $|\hat{C}\cup C| > 0.8\ell$. Then, since $|C| \leq 1.01\ell$, $|C\setminus C'| < 1.01\ell - 0.8\ell = 0.21\ell$. But then $|\hat{C}\oplus C| = |\hat{C}\setminus C| + |C\setminus \hat{C}| < 0.2\ell + 0.21\ell < \ell$. Since $\hat{C}\oplus C$ contains a cycle, and the minimum cycle length is $\ell$, this is a contradiction. Hence, $|\hat{C}\setminus C| \geq 0.2\ell$, and so the shortest cycle length in $G/C$ is $\geq 0.2\ell$.
\end{proof}

Now, we have the following claim:

\begin{claim}\label{clm:UniqueCycleSurvival}
    Let $G = (V, E)$ be a graph with shortest cycle length $\ell \leq 20\log(m)$, and let $C$ be an arbitrary cycle in $G$ of length $[\ell, 1.01 \ell]$. Consider sampling the edges of $G$ using \cref{alg:arbMarginals} with $(1/m^{500})$-almost $2\ell$-wise independence and marginals $p = (1/2)^{\left\lceil200\log(m)/\ell\right\rceil}$. Then, with probability $> 0$, the resulting sample will:
    \begin{enumerate}
        \item Sample all edges involved in the cycle $C$.
        \item Not contain any other cycle $C'$.
    \end{enumerate}
\end{claim}

\begin{proof}
    Because $C$ is a cycle of length $[\ell, 1.01 \ell]$, $\ell \leq 20\log(m)$, and $p = (1/2)^{\left\lceil 200\log(m)/\ell\right\rceil}\geq 1/(2m^{200/\ell})$, the probability $C$ survives sampling is at least:
    \[
    \Pr[C \text{ survives}] \geq \left ( \frac{1}{2m^{200/\ell}} \right )^{1.01 \ell} - \frac{1}{m^{500}} = \frac{1}{2^{1.01\ell}\cdot m^{202}} - \frac{1}{m^{500}} \geq \frac{1}{m^{250}},
    \]
    where the last inequality follows since $2^{1.01\ell} \leq m^{20.2}$ for $\ell \leq 20\log(m)$.

    At the same time, conditioned on $C$ surviving, the only way for another cycle $C'$ to survive is if the contracted graph $G / C$ contains a cycle. 

    Consider the probability distribution of edges in $G/C$ conditioned on $C$ surviving sampling. Take any event $A$ in this distribution (event depending on the edges in $E \setminus C$) that depends on $\leq 0.2\ell$ edges. Then:
\[
\Pr[A|C\text{ survives}] = \frac{\Pr[A\wedge C\text{ survives}]}{\Pr[C\text{ survives}]} = \frac{\Pr_{U}[A\wedge C\text{ survives}] + e_1}{\Pr_U[C\text{ survives}] + e_2},
\]
where $\Pr_U$ denotes the probabilities under fully independent sampling and $e_1, e_2$ are error terms. Since the event of $C$ surviving depends on $|C| \leq 1.01\ell$ edges, then both of the events above depend only on $\leq 2\ell$ edges. Our (unconditional) distribution is $\delta$-almost $2\ell$-wise independent, meaning that $|e_1|, |e_2| \leq \delta$. Moreover, for the independent distribution, $\Pr_{U}[A\wedge C\text{ survives}] = \Pr_{U}[A]\cdot \Pr_U[C\text{ survives}]$. Then note that:
\begin{equation*}
\begin{aligned}
\left|\Pr[A|C\text{ survives}] - \Pr_U[A]\right| = \left|\frac{\Pr_{U}[A]\cdot \Pr_U[C\text{ survives}]+e_1}{\Pr_U[C\text{ survives}] + e_2} - \Pr_U[A]\right| \\ \\ \leq \frac{|e_1| + \Pr_U[A]\cdot |e_2|}{\Pr[C\text{ survives}]} \leq \frac{2\delta}{\Pr[C\text{ survives}]}
\end{aligned}
\end{equation*}

Hence, since $\delta = 1/m^{500}$ and $\Pr[C\text{ survives}] \geq 1/m^{250}$, we get that the distribution of edges in $G/C$ conditioned on $C$ surviving is $(1/m^{200})$-almost $0.2\ell$-wise independent.

Since the shortest cycle length in $G / C$ is at least $\ell / 5$ (as per \cref{clm:boundedLengthContract}), by \cref{cor:smallerGirth} we know that subsampling the edges of $G/C$ using \cref{alg:arbMarginals} with $(1/m^{200})$-almost $0.2\ell$-wise independence and marginals $(1/2)^{\left\lceil 40\log(m)/0.2\ell\right\rceil}$ yields a cycle-free graph with probability $1-1/\poly(m)$\footnote{Note that the number of edges in $G/C$ is actually $m - |C| < m$; however, by using $m$ instead of $m - |C|$ we only decrease the marginal probabilities $(1/2)^{\left\lceil 40\log(m)/0.2\ell\right\rceil}$ and error term $\delta = 1/m^{500}$, meaning that cycle-freeness is also preserved with these parameters. We also note that for $\ell \leq 20\log(m)$, we have $0.2\ell \leq 20\log(m-1.01\ell)\leq 20\log(m - |C|)$ for sufficiently large $m$, and so our invocation of \cref{cor:smallerGirth} is valid.}. Therefore:
    \begin{equation*}
    \begin{aligned}
    \Pr[C \text{ survives } \wedge \not \exists C': C' \text{ survives}] = \Pr[C \text{ survives}]\cdot \Pr[\text{$G/C$ is cycle-free }|C \text{ survives}] \geq& \\ \geq \frac{1}{m^{250}}\cdot (1-1/\poly(m)) > 0&.
    \end{aligned}
    \end{equation*}
This yields the claim. 
\end{proof}

\section{Derandomizing Matroid Basis Finding}\label{sec:framework}

\subsection{Matroids}

In this section, we introduce basic properties of matroids that we will use in later sections.

\begin{definition}[Matroid]
	A matroid $\mathcal{M}$ is given by a ground set $E$ on $m$ elements, and the set of independent sets $\mathcal{I} \subseteq 2^E$. $\mathcal{I} $ satisfies three conditions:
	\begin{enumerate}
		\item $\emptyset \in \mathcal{I}$.
		\item If $S \in \mathcal{I}$ then for any $S' \subseteq S$, $S' \in \mathcal{I}$.
		\item If $S, T \in \mathcal{I}$ and $|S| > |T|$, then there exists $e \in S \setminus T$ such that $T \cup \{e\} \in \mathcal{I}$. 
		\end{enumerate}
\end{definition}

Ultimately, the algorithms we design will be used for finding \emph{matroid bases}:

\begin{definition}[Matroid Basis]
Given a matroid $\mathcal{M} = (E, \mathcal{I})$, a basis of  $\mathcal{M}$ is any set $S \subseteq E$ such that $S \in \mathcal{I}$ and $S$ is of maximal size. 
\end{definition}

\begin{definition}[Matroid Rank]
Given a matroid $\mathcal{M} = (E, \mathcal{I})$, a rank of $\mathcal{M}$ equals the cardinality of its basis.
\end{definition}

In this work, we will focus primarily on the setting of finding bases in \emph{graphic} and \emph{cographic} matroids:

\begin{definition}[Graphic Matroid]
	Given a graph $G = (V, E)$, its corresponding graphic matroid is the matroid with ground set $E$, where a set $S \subseteq E$ is independent if and only if $S$ contains no cycles in $G$.
\end{definition}

\begin{definition}[Cographic Matroid]
Given a graph $G = (V, E)$, its corresponding cographic matroid is the matroid with ground set $E$, where a set $S \subseteq E$ is independent if and only if $S$ does not completely contain any non-empty cut in $G$ (equivalently, there must be some spanning forest $F$ of $G$ such that $S \cap F = \emptyset$).
\end{definition}

\begin{remark}
	In graphic matroids, bases are spanning forests. In cographic matroids, bases are \emph{complements} of spanning forests.
\end{remark}

We will also make use of the \emph{dual} of a matroid basis:

\begin{definition}[Circuit]
	In a matroid $\mathcal{M} = (E, \mathcal{I})$, a circuit is a set $C \subseteq E$ such that $C \notin \mathcal{I}$, but for any $e \in C$, $C \setminus \{e\} \in \mathcal{I}$.
\end{definition}

\begin{remark}
	Note that in a graphic matroid, a circuit is a single, simple cycle in the underlying graph. In a cographic matroid, a circuit is the edges contained in exactly one cut in the underlying graph.
\end{remark}

We will often use the operations of \emph{deletion} and \emph{contraction} in the context of matroids:

\begin{definition}[Deletion and Contraction]
	Given a matroid $\mathcal{M} = (E, \mathcal{I})$ and a set of elements $T \subseteq E$, the result of deleting the set $T$ (denoted $\mathcal{M} \setminus T$)  is the matroid $(E \setminus T, \{S \setminus T: S \in \mathcal{I}\})$.
	
	Given a matroid $\mathcal{M} = (E, \mathcal{I})$ and a set of elements $T \subseteq E$ such that $T \in \mathcal{I}$, the result of contracting on the set $T$ (denoted $\mathcal{M} / T$)  is the matroid $(E \setminus T, \{S \setminus T: S \in \mathcal{I}: T \subseteq S\})$.
\end{definition}

\begin{remark}
	We will often use the fact that in a matroid $\mathcal{M} = (E, \mathcal{I})$, if $T \subseteq E$ and $T \in \mathcal{I}$, then for any basis $B$ of $\mathcal{M} / T$, $B \cup T$ is a basis of $\mathcal{M}$.
\end{remark}

As is typical in the matroid setting, when we work with graphic and cographic matroids, we will not assume access to the underlying graph, and instead assume access only to an \emph{independence oracle}:

\begin{definition}
	Given a matroid $\mathcal{M} = (E, \mathcal{I})$, the independence oracle $\mathrm{Ind}$ of $\mathcal{M}$ is a function which maps $2^E \rightarrow \{0,1\}$ such that for $S \subseteq E$,  $\mathrm{Ind}(S) = \mathbf{1}[S \in \mathcal{I}]$.
\end{definition}

\paragraph{Problem Description}

Our goal will be to design explicit, deterministic, parallel algorithms for finding bases of graphic and cographic matroids. These algorithms operate in a round-by-round manner, where in each round, the algorithm submits a batch of $\mathrm{poly}(m)$ queries to the independence oracle. The algorithm then reads the answers to these queries and prepares the next round of queries. The goal is to minimize both the number of rounds of queries required for finding a matroid basis and the number of queries required per round. 

\subsection{Matroid Algorithm Primitives}

Now, we recall some basic sub-routines for identifying circuits in matroids using only independence queries. First, we have the unique circuit detection algorithm as presented in \cite{khanna2025optimal}.

\begin{algorithm}[H]
	\caption{DetectSingleCircuit($E'$)}\label{alg:detectSingleCycle}
	Initialize the set of critical edges $S = \emptyset$.\\
	\If{\text{Query }$\mathrm{Ind}(E') = 1$}{
		\Return{$\perp$, No circuits.}
	}
	\For{$e \in E'$}{
		\If{\text{Query }$\mathrm{Ind}(E' - e) = 1$}{
			$S \leftarrow S \cup \{ e \}$.
		}
	}
	\If{$S = \emptyset$}{
		\Return{$\perp$, $\geq 2$ circuits.}
	}
	\Return{$S$.}
\end{algorithm}

Importantly, this algorithm satisfies the following property:

\begin{lemma}[\cite{khanna2025optimal}]\label{lem:detectSingleCycle}
	For a set of edges $E'$, \cref{alg:detectSingleCycle} returns $\perp$ if $E'$ contains $0$ or $\geq 2$ circuits, and otherwise returns $S \subseteq E'$ where $S$ is exactly the edges participating in the unique circuit in $E'$.
\end{lemma}

We also have the following building block for deleting a set of edges without altering the connectivity of the graph: 

\begin{lemma}[\cite{khanna2025optimal}]\label{lem:deleteCycles}
	Let $E$ be a set of edges with some fixed ordering of the edges $e_1, \dots e_m$, and let $\mathrm{Circuits}$ be an arbitrary subset of the circuits in $E$. For each circuit $C \in \mathrm{Circuits}$, let $C = (e_{i_{C, 1}}, \dots e_{i_{C, |C|}})$ denote the ordered set of the edges that are in the circuit $C$. Let $E'$ be the result of simultaneously deleting from $E$ the edge with the \emph{largest} index from every circuit in $\mathrm{Circuits}$. Then,
	\begin{enumerate}
		\item Every circuit in $\mathrm{Circuits}$ has at least one edge removed.
		\item The rank of $E'$ is the same as the rank of $E$. 
	\end{enumerate}
\end{lemma}

\cref{lem:deleteCycles} can be implemented in the following manner:

\begin{algorithm}[H]
	\caption{DeleteCircuits$(E, \mathrm{Circuits})$}\label{alg:deleteCycles}
	Let $e_1, \dots e_m$ be an arbitrary (but fixed) ordering of the edges in $E$. \\ 
	\For{$C \in \mathrm{Circuits}$ in parallel}{
		Let $e^*$ be the edge with the largest index in $C$. \\
		$E \leftarrow E - e^*$. \\
	}
	\Return{$E$.}
\end{algorithm}

\begin{corollary}\label{cor:deleteCycles}
	Let $G = (V, E)$ be a graph, and let $\mathrm{Circuits}$ be a set of circuits in $E$. Then, \cref{alg:deleteCycles} returns a subset of $E$ such that:
	\begin{enumerate}
		\item Every circuit in $\mathrm{Circuits}$ has at least one edge removed. 
		\item The rank of $G$ is unchanged. 
	\end{enumerate}
\end{corollary}

\begin{proof}
	\cref{alg:deleteCycles} directly implements \cref{lem:deleteCycles}.
\end{proof}

\subsection{Explicit Derandomization Framework}

Below we present our general framework that we then apply to both graphic and cographic matroids. The main idea of our algorithm is to first find and remove short circuits in a matroid, and then, when the smallest circuit is large ($> s\log n$ for some constant $s$), find a large independent set and contract the graph on it, recursively finding a basis for the contracted graph.

Specifically, suppose we have two algorithms: $\mathrm{ListShortCircuits}(E, \ell, m)$ which returns a list of circuits of size in $[\ell, 1.01\ell]$, where $\ell$ is the smallest circuit size, and $\mathrm{FindLargeIndependentSet}(E)$, which returns an independent set with $\geq |E|/10$ edges if the minimum circuit size is $>s\log m$, where $s$ is some absolute constant. Given these sub-routines, we present the main algorithm for finding a basis of a matroid:

\begin{algorithm}[H]
\caption{FindBasis$(G = (V, E), m)$}\label{alg:final}
$T = \emptyset.$ \tcp{Edges in a basis.}
\While{$E$ is not empty}{
$\ell \leftarrow 1$ \\
\While{$\ell \leq s \log(m)$}{
$\mathrm{Circuits} = \mathrm{ListShortCircuits}(G, \ell, m)$. \\
$E = \mathrm{DeleteCircuits}(E, \mathrm{Circuits})$. \\
$\ell \leftarrow \ell \cdot 1.01$.\\
}
$S = \mathrm{FindLargeIndependentSet}(E)$. \\
$E \leftarrow E / S$. \\
$T \leftarrow T \cup S$. \\
}
\Return{$T$.}
\end{algorithm}

Ultimately, we will show the following:

\begin{theorem}\label{thm:final}
    Assume we have algorithms
    
\begin{itemize}
    \item $\mathrm{ListShortCircuits}(E, \ell, m)$, which makes $Q_1$ queries and returns a set of circuits of size in $[\ell, 1.01\ell]$, where $\ell$ is the minimum circuit size, and
    \item $\mathrm{FindLargeIndependentSet}(E)$, which makes $Q_2$ queries and, provided that the minimum circuit size is $> s\log(m)$ for some absolute constant $s$, returns an independent set with $\geq |E|/10$ edges.
\end{itemize}

Then, for a graph $G$ on $m$ edges, \cref{alg:final} makes $O((Q_1 + Q_2)\operatorname{polylog}(m))$ independence queries and finds a basis of a matroid corresponding to $G$ in $O(\log(m) \log\log(m))$ rounds. 
\end{theorem}

Towards proving this theorem, we first prove the \emph{correctness} of the above algorithm:

\begin{claim}\label{clm:finalCorrect}
    For a graph $G$ on $m$ edges, \cref{alg:final} returns a basis of a matroid corresponding to $G$. 
\end{claim}

\begin{proof}
We prove this via the inductive claim that the set $T$ always corresponds to a set of independent edges in $E$, and that $T \cup E$ always maintains the same matroid rank. Within any given iteration, observe that $\mathrm{ListShortCircuits}(E, \ell, m)$ returns a valid set of circuits by definition, and so by \cref{cor:deleteCycles}, $\mathrm{DeleteCircuits}(E, \mathrm{Circuits})$ does not alter the rank of $E$ (and of $T\cup E$).

Thus, the only modification to the rank of $E$ happens when we contract on the set $S$ recovered by $\mathrm{FindLargeIndependentSet}$. By definition, $\mathrm{FindLargeIndependentSet}$ always recovers independent edges, and so $E$ is always contracted on a set of independent edges which are subsequently added to $T$. 

Note that if $S$ is independent, and $S'$ is a basis of $E / S$, then $S \cup S'$ is a basis of $E$. Thus, we show by induction that in any given iteration, $T$ is exactly the set which has been contracted on, and the remainder of the algorithm finds a basis of $E / T$, which is added to the set $T$. Therefore, $T$ will indeed constitute a basis in $E$. 
\end{proof}

\begin{claim}\label{clm:finalRoundBound}
    For a graph $G$ on $m$ edges, \cref{alg:final} makes $O((Q_1 + Q_2)\operatorname{polylog}(m))$ queries, and terminates in $O(\log(m)\log\log(m))$ rounds. 
\end{claim}

\begin{proof}
    The bound on queries follows trivially by the individual query bounds for every subroutine, $\mathrm{ListShortCircuits}(E, \ell, m)$ and $\mathrm{FindLargeIndependentSet}(E)$ along with the assumed bound on the number of rounds.

    To see the bound on the number of rounds, observe that the while loop on line $4$ always runs for a maximum of $O(\log\log(m))$ rounds. This is because in each interior iteration, $\ell \rightarrow 1.01\ell$, and so after $O(\log\log(m))$ rounds, $\ell > s \log(m)$.

Next, we can observe that the outer while loop (line $2$) only runs for $O(\log(m))$ rounds, as each invocation of $\mathrm{FindLargeIndependentSet}(E)$ recovers $\geq |E| / 10$ independent edges which are subsequently contracted on (decreasing the number of edges in $E$ by a constant factor). Thus, after $O(\log(m))$ rounds of the outer while loop, $E$ becomes empty. 

Thus, in total the number of rounds is $O(\log(m)\log\log(m))$.
\end{proof}

Now, we are ready to conclude \cref{thm:final}.

\begin{proof}[Proof of \cref{thm:final}]
    The correctness of the algorithm follows from \cref{clm:finalCorrect}, and the bound on the number of queries and rounds follows from \cref{clm:finalRoundBound}.
\end{proof}

\section{Near-Optimal Explicit Derandomization of Basis Finding in Graphic Matroids}\label{sec:graphicDerand}

In this section, we proceed by instantiating the sub-routines required by \cref{alg:final} for graphic matroids. Specifically, we give an efficient algorithm for finding short cycles, and finding a large independent set given that the shortest cycle length is large.

\subsection{Finding a Large Independent Set}
We first show that if the shortest cycle length is sufficiently large, then there is an explicit set of queries that is guaranteed to find a large independent set. This follows from \cref{clm:derandomizeSampling}:
\begin{lemma}\label{lem:deterministicQuery}
    Let $G = (V, E)$ be a graph with $m$ edges and shortest cycle length $> 20 \log(m)$. Then, there is an explicit set of $\mathrm{poly}(m)$ queries $Q \subseteq 2^E, |Q| = \mathrm{poly}(m)$, which guarantees that there is some $\hat{E} \in Q$ for which:
    \begin{enumerate}
        \item $|\hat{E}| \geq \frac{m}{10}$.
        \item $\mathrm{Ind}(\hat{E}) = 1$.
    \end{enumerate}
\end{lemma}

\begin{proof}
    The set $Q$ that we use is the entire range of the distribution of \cref{clm:derandomizeSampling}. Importantly, \cref{thm:boundedIndSampling} implies that the support of this distribution (i.e. the size of $Q$) is bounded in size by:
    \[
    2^{O(\log(m) + \log\log(m) + \log(\mathrm{poly}(m))} = \mathrm{poly}(m).
    \]
    \cref{clm:derandomizeSampling} shows that one such sample from $Q$ will satisfy both of the stated conditions above. 
\end{proof}

Importantly, \cref{lem:deterministicQuery} allows us to implement $\mathrm{FindLargeIndependentSet}$:

\begin{algorithm}[H]
    \caption{FindLargeIndependentSet$(G, m)$ (graphic matroids)} \label{alg:findLargeIndependentSet}
    Let $Q$ be the set of queries guaranteed by \cref{lem:deterministicQuery}. \\
    $S^* = \emptyset$. \\
    \For{$S \in Q$}{
    \If{$\mathrm{Ind}(S) = 1 \wedge |S| \geq |S^*|$}{
    $S^* = S$. 
    }
    }
    \Return{$S^*$.}
\end{algorithm}

We have the following guarantee on the above algorithm:

\begin{corollary}\label{cor:findLargeIndependentSet}
    For a graph $G$ on $m$ edges with minimum cycle length $\ell \geq 20 \log(m)$, \cref{alg:findLargeIndependentSet} makes $\mathrm{poly}(m)$ queries and finds a set of $\geq m/10$ edges with no cycles. 
\end{corollary}

\begin{proof}
    For such a graph, by \cref{lem:deterministicQuery}, we know that the set of queries $Q$ contains some set of edges $\hat{E}$ for which $\mathrm{Ind}(\hat{E}) = 1$, and $|\hat{E}| \geq m/10$. \cref{alg:findLargeIndependentSet} queries all the sets in $Q$, and returns the independent set of largest size, and so the returned set is of size $\geq m/10$. The bound on the number of queries follows because $|Q|  = \mathrm{poly}(m)$.
\end{proof}

\subsection{Removing Short Cycles}
\begin{lemma}\label{lem:listShortCycles}
    Let $G = (V, E)$ be a graph with $m$ edges and shortest cycle length $\ell \leq 20 \log(m)$. Then, there is an explicit set of $\mathrm{poly}(m)$ queries $Q \subseteq 2^E, |Q| = \mathrm{poly}(m)$, which guarantees that for every cycle $C$ in $G$ of length $[\ell, 1.01 \ell]$, there is some $\hat{E} \in Q$ for which $C$ is the unique cycle contained in $\hat{E}$.
\end{lemma}

\begin{proof}
    First, if $\ell \leq 100$, then we simply query all sets of $\leq 101$ edges. By construction, this only requires $\mathrm{poly}(m)$ queries, and necessarily, for any cycle $C$ of length $\leq 101$, there will be some query which contains exactly $C$. 

    Otherwise, for $\ell > 100$, we let $S$ be the support the entire support of a $(1/m^{250\gamma})$-almost $2\ell$-wise independent distribution with marginals $p = (1/2)^{\left\lceil 10\gamma\log(m)/\ell\right\rceil}$ sampled by \cref{alg:arbMarginals} with the input distribution obtained from \cref{thm:boundedIndSampling}. By \cref{clm:UniqueCycleSurvival} we know that for any cycle $C$ of length $[\ell, 1.01\ell]$, there must exist a query from $S$ which makes $C$ the unique surviving cycle. Also, the distribution sampled by \cref{alg:arbMarginals} with \cref{thm:boundedIndSampling}, $p = (1/2)^{\Theta(\log(m)/\ell)}$, $\delta = 1/\poly(m)$, and $k = O(\ell)$ has the following support size:
    \[
    2^{O(k \log(1/p) + \log\log((m \log(1/p)) \log(1/p)) + \log(1 / \delta))} = 2^{O\left(\ell \cdot \frac{\log(m)}{\ell}  + \log(m)\right)} = 2^{O(\log(m))} = \mathrm{poly}(m).
    \]
    Importantly, observe that $\log(1/p) = \left\lceil200\log(m)/\ell\right\rceil$ is an integer, and thus our invocation of \cref{alg:arbMarginals} is valid. This finishes the lemma. 
\end{proof}

With this, we present the following implementation of $\mathrm{ListShortCircuits}$:

\begin{algorithm}[H]
    \caption{ListShortCycles$(G, \ell, m)$ (graphic matroids)}\label{alg:ListShortCycles}
    Let $Q$ be the set of queries guaranteed by \cref{lem:listShortCycles} with parameters $\ell, m$. \\
    $\mathrm{Cycles} = \emptyset$. \\
    \For{$S \in Q$}{
    \If{DetectSingleCycle$(S) \neq \perp$}{
    $\mathrm{Cycles} \leftarrow \mathrm{Cycles}  \cup \mathrm{DetectSingleCycle}(S)$.
    }
    }
    \Return{$\mathrm{Cycles}$.}
\end{algorithm}

We immediately have the following guarantee on the performance of ListShortCycles (\cref{alg:ListShortCycles}):

\begin{corollary}\label{cor:listShortCycles}
    For a graph $G$ on $\leq m$ edges with minimum circuit length $\ell$, \cref{alg:ListShortCycles} makes $\poly(m)$ queries and returns a set $\mathrm{Cycles}$ that contains all cycles of length $[\ell, 1.01\ell]$.
\end{corollary}

\begin{proof}
    We know by \cref{lem:listShortCycles} that for every cycle $C$ of length $[\ell, 1.01\ell]$, there is some query $S \in Q$ for which $C$ is the unique cycle contained in $S$. By \cref{lem:detectSingleCycle}, whenever this is the case, $\mathrm{DetectSingleCycle}(S)$ will exactly recover the identity of the contained cycle, after which it is then added to the set $\mathrm{Cycles}$. The bound on the number of queries follows from \cref{lem:listShortCycles}, which states that $|Q| = \poly(m)$.
\end{proof}

\subsection{Complete Algorithm}
A uniform derandomization algorithm in graphic matroids follows from the correctness of the above sub-routines and the general framework established in \cref{sec:framework}:
\begin{theorem}\label{thm:finalGraphic}
For a graph $G$ on $m$ edges, \cref{alg:final} with \cref{alg:findLargeIndependentSet} and \cref{alg:ListShortCycles} as sub-routines returns a spanning forest of $G$ using $\poly(n)$ queries in $O(\log(m) \log\log(m))$ rounds. 
\end{theorem}
\begin{proof}
\cref{cor:findLargeIndependentSet} and \cref{cor:listShortCycles} show that \cref{alg:findLargeIndependentSet} and \cref{alg:ListShortCycles} satisfy the requirements of \cref{thm:final} with $Q_1 = \poly(n)$ and $Q_2 = \poly(n)$.
\end{proof}

\section{Explicit Derandomization of Basis Finding in Cographic Matroids}\label{sec:cographicDerand}

Below we present the implementations of sub-routines required by \cref{thm:final} for cographic matroids.

\subsection{Finding a Large Independent Set}

\begin{lemma}\label{lem:deterministicQueryCographic}
Let $G = (V, E)$ be a graph with $m$ edges and the minimum cut $\geq \kappa\log n$ (with $\kappa$ as in \cref{cor:connectedness-distribution-dnfs}). Then, there is an explicit set of $m^{O(\log\log(m))}$ queries $Q\subseteq 2^E$, $|Q| = m^{O(\log\log(m))}$, which guarantees that there is some $\hat{E} \in Q$ for which:
    \begin{enumerate}
        \item $|\hat{E}| \geq \frac{m}{10}$.
        \item $\mathrm{Ind}(\hat{E}) = 1$.
    \end{enumerate}
\end{lemma}

\begin{proof}
By \cref{cor:connectedness-distribution-dnfs}, there exists an explicit distribution $Y$ on $\{0, 1\}^m$ with marginals $\leq 1/2$ such that for a graph with minimal cut $> \kappa\log m$, subsampling the edges using $Y$ yields a connected graph with probability $1 - 1/\poly(m)$. In particular, this means that this graph has non-zero intersection with any cut, meaning that the edges in its complement form an independent set, i.e. do not fully contain any cut. Hence, taking the distribution $Y'$ whose samples are complements of samples in $Y$, we ensure that sampling from $Y'$ gives an independent set with probability $1 - 1/\poly(m)$.

Moreover, since the marginal probabilities of $Y$ are $\leq 1/2$, then the marginals for $Y'$ are $\geq 1/2$, meaning that by Markov's inequality, with probability $\Omega(1)$ $Y'$ samples $\geq m/10$ edges. Taking the union bound of this event and the connectedness events gives that with positive probability, a sample from $Y'$ is both connected and contains $\geq m/10$ edges. 

Therefore, we let $Q$ to be the entire support of $Y'$. By \cref{cor:connectedness-distribution-dnfs}, $|Q| = m^{O(\log\log(m))}$.

\end{proof}

Similarly to the case of graphic matroids, this allows us to implement $\mathrm{FindLargeIndependentSet}$:

\begin{algorithm}[H]
    \caption{FindLargeIndependentSet$(G, m)$ (cographic matroids)} \label{alg:findLargeIndependentSetCographic}
    Let $Q$ be the set of queries guaranteed by \cref{lem:deterministicQueryCographic}. \\
    $S^* = \emptyset$. \\
    \For{$S \in Q$}{
    \If{$\mathrm{Ind}(S) = 1 \wedge |S| \geq |S^*|$}{
    $S^* = S$. 
    }
    }
    \Return{$S^*$.}
\end{algorithm}

We have the following guarantee on the above algorithm:

\begin{corollary}\label{cor:findLargeIndependentSetCographic}
    For a graph $G$ on $m$ edges with minimum cut $\ell \geq \kappa\log(m)$ (with $\kappa$ as in \cref{cor:connectedness-distribution-dnfs}), \cref{alg:findLargeIndependentSetCographic} makes $m^{O(\log\log(m))}$ queries and finds a set of $\geq m/10$ edges with no cut completely included in the set. 
\end{corollary}

\begin{proof}
    Follows trivially from \cref{lem:deterministicQueryCographic}.
\end{proof}

\subsection{Removing Small Cuts}

\begin{lemma}\label{lem:listSmallCuts}
    Let $G = (V, E)$ be a graph with $m$ edges and minimum cut size $\ell \leq \kappa\log(m)$ (for $\kappa$ as in \cref{cor:connectedness-distribution-dnfs}). Then, there is an explicit set of $m^{O(\log\log(m))}$ queries $Q \subseteq 2^E, |Q| = m^{O(\log\log(m))}$, which guarantees that for every cut $C$ in $G$ of size $[\ell, 1.01 \ell]$, there is some $\hat{E} \in Q$ for which $C$ is the unique cut contained in $\hat{E}$.
\end{lemma}

\begin{proof}
    We let $Q$ to be the entire support of the distribution $Y$ given by \cref{thm:unique-survival-from-dnfs}.  By \cref{thm:unique-survival-from-dnfs}, we know that for any cut $C$ of size $[\ell, 1.01\ell]$, $C$ will be the unique surviving cycle under this distribution with probability $> 0$. This means there must exist a query from the support of this distribution which makes $C$ the unique surviving cut. Moreover, the same theorem implies that $|Q| = m^{O(\log\log(m))}$.
\end{proof}

With this, we present the following implementation of $\mathrm{ListSmallCuts}$:

\begin{algorithm}[H]
    \caption{ListSmallCuts$(G, \ell, m)$ (cographic matroids)}\label{alg:ListSmallCuts}
    Let $Q$ be the set of queries guaranteed by \cref{lem:listSmallCuts} with parameters $\ell, m$. \\
    $\mathrm{Cuts} = \emptyset$. \\
    \For{$S \in Q$}{
    \If{DetectSingleCircuit$(S) \neq \perp$}{
    $\mathrm{Cuts} \leftarrow \mathrm{Cuts}  \cup \mathrm{DetectSingleCircuit}(S)$.
    }
    }
    \Return{$\mathrm{Cuts}$.}
\end{algorithm}

We immediately have the following guarantee on the performance of ListSmallCuts (\cref{alg:ListSmallCuts}):

\begin{corollary}\label{cor:listSmallCuts}
    For a graph $G$ on $\leq m$ edges with minimum cut size $\ell$, \cref{alg:ListSmallCuts} makes $m^{O(\log\log(m))}$ queries and returns a set $\mathrm{Cuts}$ that contains all cuts of size $[\ell, 1.01\ell]$.
\end{corollary}

\begin{proof}
    We know by \cref{lem:listSmallCuts} that for every cut $C$ of size $[\ell, 1.01\ell]$, there is some query $S \in Q$ for which $C$ is the unique cut contained in $S$. By \cref{lem:detectSingleCycle}, whenever this is the case, $\mathrm{DetectSingleCircuit}(S)$ will exactly recover the identity of the contained cut, after which it is then added to the set $\mathrm{Cuts}$. The bound on the number of queries follows from \cref{lem:listSmallCuts}, which states that $|Q| = m^{O(\log\log(m))}$.
\end{proof}

\subsection{Complete Algorithm}
Explicit derandomization algorithm in cographic matroids follows from the correctness of the above sub-routines and the general framework established in \cref{sec:framework}:
\begin{theorem}\label{thm:finalCoGraphic}
For a graph $G$ on $m$ edges, \cref{alg:final} with \cref{alg:findLargeIndependentSetCographic} and \cref{alg:ListSmallCuts} as sub-routines returns a spanning forest of $G$ using $m^{O(\log\log(m))}$ queries in $O(\log m\cdot \log\log m)$ rounds.
\end{theorem}
\begin{proof}
\cref{cor:findLargeIndependentSetCographic} and \cref{cor:listSmallCuts} show that \cref{alg:findLargeIndependentSetCographic} and \cref{alg:ListSmallCuts} satisfy the requirements of \cref{thm:final} and the desired query bounds.
\end{proof}

\section{Acknowledgments}
The authors would like to thank CCC reviewers and Madhu Sudan for helpful points on presentation.

\bibliographystyle{alpha}
\bibliography{references}

\newpage 

\begin{appendices}
\section{Unique Cut Survival Under Bounded-Independence Edge Subsampling}\label[appendix]{appendix:unique-survival}

First, we show that \cref{cor:connectednessBoundedIndSampling} can be extended for graphs with smaller minimum cut:

\begin{corollary}\label{cor:smallerMinCutConnected}
    Given an undirected unweighted graph $G = (V, E)$ on $m$ edges, if the minimum cut of $G$ is $\ell \leq \kappa \log m$ (for $\kappa$ as in \cref{cor:connectednessBoundedIndSampling}), then sampling the edges of $G$ using \cref{alg:arbMarginals} with $p = 1 - (1/2)^{\left\lceil2\kappa \log m/\ell\right\rceil}$ and $O(\ell)$-wise independence yields a connected graph with probability $1 - 1 / \mathrm{poly}(m)$.
\end{corollary}

\begin{proof}
    We construct an auxiliary graph $G'$, whereby we replace each edge in $G$ with $s = \left\lceil2\kappa \log(m)/\ell\right\rceil$ many multi-edges (again, using $\kappa$ to be the constant from \cref{cor:connectednessBoundedIndSampling}). Immediately, we can see that the minimum cut in $G'$ is of size $\ell' \geq 2\kappa \log(m)$. At the same time, the number of edges in $G'$ is $m' = sm \leq m^2$, since $s \leq m$ for sufficiently large $m$. Hence, we get that $\ell' \geq 2\kappa\log(m) \geq \kappa\log(m')$, meaning that we can apply \cref{cor:connectednessBoundedIndSampling} to $G'$. Namely, if one samples the edges of $G'$ using $O(\log(m))$-wise independent random bits with marginals $1/2$, the resulting graph is connected with probability $1 - 1 / \mathrm{poly}(m)$.

    Now, observe that our implementation of $O(\ell)$-wise independent sampling with marginals $(1/2)^{\left\lceil 2\kappa \log(n)/\ell\right\rceil}$ from \cref{alg:arbMarginals} exactly corresponds to the above approach.
\end{proof}

Using this, we can show that there is a non-negligible probability that \emph{any} near-minimum cut is the \emph{unique} surviving cut in a random $k$-wise independent sample of edges.

More formally, we have the following claim:

\begin{claim}\label{clm:uniqueCutSurvival}
    Let $G = (V, E)$ be a graph with $m$ edges and minimum (non-empty) cut size $\ell \leq \kappa\log m$ (for $\kappa$ as in \cref{cor:connectednessBoundedIndSampling}), and let $C \subseteq E$ denote the edges participating in some cut of size $[\ell, 1.01 \ell]$. Now, consider sampling the edges of $G$ using \cref{alg:arbMarginals} with $O(\ell)$-wise independence and marginals $p = (1/2)^{\left\lceil10\kappa\log(m)/\ell\right\rceil}$, with probability $>0$, the resulting sample $Q \subseteq E$ will both:
    \begin{enumerate}
        \item Contain $C$, i.e., $C \subseteq Q$.
        \item For all other (non-empty) cuts $C' \neq C$, $C' \not \subset Q$. 
    \end{enumerate}
\end{claim}

Note that this claim implies that for each near-minimum cut $C$, it will be the \emph{unique} surviving cut in some sample from our distribution. Before proving this claim, we first require the following structural characterization of cut sizes in graphs under edge removals:

\begin{claim}\label{clm:removeCutNewCutSize}
    Let $G = (V, E)$ be a graph with minimum (non-empty) cut $\ell$, and let $C \subseteq E$ denote the edges participating in a cut of size $[\ell, 1.01 \ell]$. Now, let $G - C$ denote the result of deleting all the edges in $C$ from $G$. Then:
    \begin{enumerate}
        \item $G - C$ has one more connected component than $G$.
        \item The minimum (non-empty) cut in $G - C$ is of size $\geq 0.2 \ell$.
    \end{enumerate}
\end{claim}

\begin{proof}
    WLOG, we assume that the graph $G$ is connected (i.e., there is only one connected component). Note that this is WLOG as if $G$ has more than one connected component, any near-minimum cut must still be completely contained within one of these components (otherwise its size would be $\geq 2\ell$).
    
    Beyond this, the key starting observation is that the cuts in a graph form an $\mathbb{F}_2$-vector space over $\mathbb{F}_2^E$ \cite{Oxl06}. That is to say, for any two cuts $C, C'$ in a graph $G$, it will be the case that $C \oplus C'$ is also a cut in the graph $G$.

    Consider a cut $C' \neq C$ in $G$. We will show that $|C' \setminus C| \geq 0.2 \ell$. Assume that $|C'\setminus C| < 0.2\ell$, then, since $|C'| \geq \ell$ ($\ell$ is the minimum cut size), it must be that $|C'\cap C| > 0.8\ell$. However, since $|C| \leq 1.01\ell$, then $|C\setminus C'| = |C| - |C\cap C'| < 1.01\ell - 0.8\ell = 0.21\ell$. But this means that for the cut $C\oplus C'$, $|C\oplus C'| = |C'\setminus C|+|C\setminus C'| < 0.2\ell + 0.21\ell < \ell$, which contradicts $\ell$ being the minimum cut. Hence, $|C'\setminus C| \geq 0.2\ell$.
        
        In particular, $C'\setminus C \neq \emptyset$, so $C'\not\subseteq C$. This implies that $G - C$ has \textit{exactly} two connected components: it has at least two because we removed all edges in a cut $C$ (disconnecting two of its sides), while it is at most two because all other cuts $C'\neq C$ have at least one edge preserved. This proves the first part of the claim.

        Finally, the fact that $|C'\setminus C| \geq 0.2\ell$ for all cuts $C'\neq C$ also clearly implies that the minimum (non-empty) cut in $G - C$ is of size at least $0.2\ell$, proving the second part of the claim.
\end{proof}

We now prove \cref{clm:uniqueCutSurvival}.

\begin{proof}[Proof of \cref{clm:uniqueCutSurvival}.]
    Let $k \geq 1.01\ell$ be an independence parameter to be chosen later.
    
    First, we analyze the probability that the cut $C$ survives under sampling. Because $|C| \leq 1.01\ell$ and we sample with marginals $p = (1/2)^{\left\lceil10\kappa\log(m)/\ell\right\rceil} \geq 1/(2m^{10\kappa/\ell})$ in a $k$-wise independent manner, we see that 
    \[
    \Pr[C \text{ survives}] = p^{|C|} \geq \left ( \frac{1}{2m^{10\kappa/\ell}}\right )^{1.01\ell} \geq \frac{1}{2^{1.01\ell}\cdot m^{10.1\kappa}} \geq \frac{1}{m^{12\kappa}},
    \]
    where we have used that $|C| < 1.01\ell \leq k$, and so its survival probability under $k$-wise independent sampling is the same as its survival probability under uniform sampling. 

    Next, we consider the probability that any cut $C' \neq C$ \emph{also} survives when sampling at rate $p$ conditioned on $C$ surviving sampling. To analyze this, observe that conditioned on $C$ surviving sampling, a cut $C'$ survives sampling if and only if all edges in $C' \setminus C$ survive sampling. Thus, in order to argue that \emph{no other} cut $C'$ survives sampling conditioned on $C$ surviving sampling, it suffices to argue that in the graph $G - C$, there is no cut that completely survives sampling. Formally,
    \[
    \Pr[\not \exists C' \neq C: C' \text{ survives sampling} | C \text{ survives sampling}]
    \]
    \[
    \geq \Pr[\not \exists C': C' \text{ survives sampling in }G - C| C \text{ survives sampling}].
    \]
     Let us denote the sampled edges in $G - C$ by $\widehat{G - C}$. No cut surviving sampling in $G-C$ is equivalent to $(G-C) - (\widehat{G - C})$ having the same number of connected components as $(G-C)$.

     Since the original distribution is $O(\ell)$-wise independent, if we condition it on $C$ surviving sampling, the remaining distribution will be $O(\ell)$-wise independent since $|C| \leq 1.01\ell$. Moreover, it is easy to verify that the conditional distribution is still sampled by \cref{alg:arbMarginals} with marginals $p = (1/2)^{\left\lceil 2\kappa\log(m)/0.2\ell\right\rceil}$.

    Finally, we note that \cref{thm:connectednessBoundedIndSampling} and \cref{cor:smallerMinCutConnected} can be directly generalized for disconnected graphs with large minimum non-empty cut, in which case we get that the connected components are preserved with high probability\footnote{Formally, if the minimal non-empty cut in a disconnected graph $c$ is at least $\kappa \log(m)$, for $\kappa$ as in \cref{thm:connectednessBoundedIndSampling}, one can still find a weighting such that the graph is sparsified when sampling with marginals at most $1/2$ in a $k$-wise independent manner. To get this weighting, we simply invoke \cref{thm:lev-scores-min-cut} separately for each connected component. Since graph sparsifiers preserve connected components, we get that connectivity is preserved with high probability. \cref{cor:smallerMinCutConnected} can be applied to this generalized statement of the theorem without any change.}. By \cref{clm:removeCutNewCutSize}, we know that the minimum non-empty cut in $G - C$ is of size $\geq 0.2\ell$. Thus, in order to argue that removing $\widehat{G - C}$ from $G - C$ does not change connected components, we must only show that $O(\ell)$-wise independent sampling in a graph with minimum non-empty cut $\geq 0.2 \ell$ with marginals $1 - (1/2)^{\left\lceil 2\kappa\log(m)/0.2\ell\right\rceil}$ does not change connected components with high probability, since this is exactly the distribution of $(G-C) - (\widehat{G - C})$. By invoking \cref{cor:smallerMinCutConnected}\footnote{Again, we note that the corollary holds even for disconnected graphs with large non-empty minimum cut.}, we get that this latter claim follows for some $k-1.01\ell = O(\ell)$, meaning that we can set $k = O(\ell)$. In this case, $\Pr[\not \exists C': C' \text{ survives sampling in }G - C | C \text{ survives sampling}] \geq 1 - 1/\poly(m)$.

    To conclude, we see that
    \[
    \Pr[C \text{ unique surviving cut}] 
    \]
    \[
    =  \Pr[C \text{ survives}] \cdot \Pr[\not \exists C' \neq C: C' \text{ survives sampling} | C \text{ survives sampling}]
    \]
    \[
    \geq \frac{1}{m^{12\kappa}} \cdot (1 - 1 / \mathrm{poly}(m)) \geq \frac{1}{2m^{12\kappa}} > 0.
    \]
\end{proof}
\end{appendices}

\end{document}

%% file: thmmacros.tex
\usepackage{amsthm}
\usepackage{thmtools,thm-restate}

\numberwithin{equation}{section}
\declaretheoremstyle[bodyfont=\it,qed=\qedsymbol]{noproofstyle}

\declaretheorem[numberlike=equation,name=Observation, Refname={Observation,Observations}]{observation}

\declaretheorem[name=Observation,numbered=no]{observation*}

\declaretheorem[numberlike=equation]{fact}

\declaretheorem[numberlike=equation]{theorem}

\declaretheorem[name=Theorem,numbered=no]{theorem*}

\declaretheorem[numberlike=equation]{lemma}
\declaretheorem[name=Lemma,numbered=no]{lemma*}

\declaretheorem[numberlike=equation]{corollary}
\declaretheorem[name=Corollary,numbered=no]{corollary*}

\declaretheorem[name=Proposition,numbered=no]{proposition*}

\declaretheorem[numberlike=equation,name=Claim, Refname={Claim,Claims}]{claim}
\declaretheorem[name=Claim,numbered=no]{claim*}

\declaretheorem[name=Conjecture,numbered=no]{conjecture*}

\declaretheorem[name=Question,numbered=no]{question*}

\declaretheoremstyle[bodyfont=\it,qed=$\lrcorner$]{defstyle} 

\declaretheorem[numberlike=equation,style=defstyle]{definition}
\declaretheorem[unnumbered,name=Definition,style=defstyle]{definition*}

\declaretheorem[unnumbered,name=Example,style=defstyle]{example*}

\declaretheorem[unnumbered,name=Notation=defstyle]{notation*}

\declaretheorem[unnumbered,name=Construction,style=defstyle]{construction*}

\declaretheorem[numberlike=equation,style=defstyle]{remark}
\declaretheorem[unnumbered,name=Remark,style=defstyle]{remark*}
